\documentclass[12pt]{article}
\usepackage[margin = 1in,bottom=1.7in]{geometry} 

\newcommand{\blind}{0} 

\usepackage{amsmath, amssymb, amsthm, bbm, mathabx}
\usepackage{authblk, booktabs, tabularx, graphicx, textcomp, url, rotating}
\usepackage[round]{natbib}
\usepackage[colorlinks,citecolor=blue, pageanchor=false]{hyperref}
\usepackage{xcolor}

\theoremstyle{plain}  
\newtheorem{thm}{Theorem}[section] 
\newtheorem{lem}[thm]{Lemma} 
\newtheorem{prop}[thm]{Proposition}

\theoremstyle{definition}

\theoremstyle{remark}

\newcommand{\st}{\preceq_{\rm st}}

\newcommand{\by}{\pmb{y}}

\newcommand{\bvartheta}{\pmb{\vartheta}}

\newcommand{\bhF}{\pmb{\hat{F}}}
\newcommand{\bhq}{\pmb{\hat{q}}}
\newcommand{\diff}{\, {\rm d}}

\newcommand{\one}{\mathbbm{1}}
\newcommand{\crps}{\operatorname{\mathrm{CRPS}}}

\DeclareMathOperator*{\argmin}{argmin}

\newcommand{\beginsupplement}{%
        \setcounter{table}{0}
        \renewcommand{\thetable}{S\arabic{table}}%
        \setcounter{figure}{0}
        \renewcommand{\thefigure}{S\arabic{figure}}%
     }

\newcommand{\orange}{\color{black}}

\begin{document}
\title{Distributional (Single) Index Models}
\if1\blind
{
} \fi
\if0\blind
{
  	\author{Alexander Henzi, Gian-Reto Kleger, Johanna F.~Ziegel\thanks{Alexander Henzi is PhD student in Statistics, Johanna F.~Ziegel is Professor of Applied Stochastics, Institute of Mathematical Statistics and Actuarial Science, University of Bern, Alpeneggstrasse 22, 3012 Bern, Switzerland (e-mail: \url{alexander.henzi@stat.unibe.ch}, \url{johanna.ziegel@stat.unibe.ch}); Gian-Reto Kleger, MD, is Head the Division of Intensive Care Medicine, Cantonal Hospital, St.~Gallen, Rorschacherstrasse 95, 9007 St.Gallen, Switzerland
(e-mail: \url{gian-reto.kleger@kssg.ch}). A.~Henzi and J.~F.~Ziegel gratefully acknowledge financial support from the Swiss National Science Foundation. The authors thank the Swiss Society of Intensive Care Medicine for providing the data. The work has greatly benefitted from discussions with Lutz D\"umbgen, Tilmann Gneiting, Alexander Jordan, and Alexandre M\"osching. This is an Accepted Manuscript of an article published by Taylor \& Francis in the Journal of the American Statistical Association  on 19 July 2021, available online: \url{https://www.tandfonline.com/doi/full/10.1080/01621459.2021.1938582}}}
} \fi
\date{\today}
\maketitle

\begin{abstract}
A Distributional (Single) Index Model (DIM) is a semi-parametric model for distributional regression, that is, estimation of conditional distributions given covariates. The method is a combination of classical single index models for the estimation of the conditional mean of a response given covariates, and isotonic distributional regression. The model for the index is parametric, whereas the conditional distributions are estimated non-parametrically under a stochastic ordering constraint. We show consistency of our estimators and apply them to a highly challenging data set on the length of stay (LoS) of patients in intensive care units. We use the model to provide skillful and calibrated probabilistic predictions for the LoS of individual patients, that outperform the available methods in the literature.

\end{abstract}

\newpage
\section{Introduction}
Regression approaches for the full conditional distribution of an outcome given covariates are gaining momentum in the literature \citep[][{\orange and the references therein}]{Hothorn2014}. They have already become an indispensable tool in probabilistic weather forecasting \citep{Gneiting2014, Vannitsem2018} but also find numerous applications in other fields such as economics, social sciences and medicine; see e.g.~\cite{Machado2000}, \cite{Chernozhukov2013}, \cite{Klein2015}, \cite{Duarte2017} and \cite{Silbersdorff2018}.

If the outcome is real-valued, then conditional distributions can be characterized in terms of their cumulative distribution function (CDF) or quantile function, and various techniques for the estimation of these objects have been proposed. {\orange \cite{Foresi1995} and} \cite{Peracchi2002} {\orange build} on the extant methods for the estimation of single quantiles or probabilities \citep{Koenker2005}, and suggest to approximate the conditional distribution by a cascade of regressions for quantiles or for the CDF evaluated at certain thresholds. A drawback of this approach is that the resulting estimates are not necessarily isotonic (the so-called 'quantile crossing problem') and thus require correction, for which remedies have already been developed, see e.g.~\cite{DetteVolgushev2008,Chernozhukov2010}.

A broad class of methods that directly yield well-defined probability distributions are generalized additive models for location, shape and scale \citep[GAMLSS]{Rigby2005}. They build on generalized linear models \citep[GLM]{McCullagh1989} and generalized additive models for the mean \citep[GAM]{Hastie1990} but also allow to model shape and scale parameters as functions of covariates. The GAMLSS framework has has been extended to Bayesian statistics \citep{Umlauf2018} and combined with popular machine learning techniques such as boosting \citep{Thomas2018}, neural networks \citep{Rasp2018} and regression forests \citep{Schlosser2019}.

Finally, there are also powerful semi-parametric and nonparametric techniques for the estimation of conditional distributions. Fully nonparametric methods estimate the conditional distribution functions locally, for example by kernel functions \citep{Hall1999, Dunson2007, Li2008}, or by partitioning of the covariate space, as in quantile random forests \citep{Meinshausen2006, Athey2019}.
A frequently used semi-parametric distributional regression method is Cox regression \citep{Cox1972}, which models the hazard rate of the outcome but also allows to derive its survival function. Conditional transformation models \citep{Hothorn2014} assume a parametric distribution for an unknown monotone transformation of the response, which is estimated along with the model parameters. \citet{Hall2005,Zhang2017} propose semi-parametric methods that reduce the dimension of the covariate space by a suitable projection, and then estimate the conditional distributions non-parametrically given the projections by kernel methods. 

We introduce a new approach to distributional regression that can be seen as a combination of a single index model with isotonic distributional regression \citep[IDR,][]{Henzi2019}. The dimension reduction of the covariate space achieved by the single index assumption is in the spirit of \citet{Hall2005,Zhang2017} but the combination with IDR is new, and has the advantage to be free of any implementation choices or tuning parameters. 

Let $Y$ be a real-valued response and $X$ a covariate in some covariate space $\mathcal{X}$. We want to estimate the conditional distribution of $Y$ given $X$, that is, $\mathcal{L}(Y\mid X)$. To expose the main idea, suppose that $\mathcal{X} = \mathbb{R}^d$. Then, a Distributional (Single) Index Model (DIM) could be
\begin{equation}\label{eq:SIM}
\mathbb{P}(Y \le y \mid X) = F_{\alpha_0^\top X}(y), \quad \text{for all $y \in \mathbb{R}$,}
\end{equation}
where $\alpha_0 \in \mathbb{R}^d$, $\alpha_0^\top X$ denotes the scalar product between $\alpha_0$ and $X$, and $(F_u)_{u \in \mathbb{R}}$ is a family of CDFs such that 
\begin{equation}\label{eq:iso}
F_u \st F_v \quad \text{if $u \le v$},
\end{equation}
where $\st$ denotes the usual stochastic order, that is $F_u \st F_v$ if $F_u(y) \ge F_v(y)$ for all $y \in \mathbb{R}$. We call $\theta(x) = \alpha_0^Tx$ in representation (\ref{eq:SIM}) the index (function). 

If the parameter $\alpha_0$ in the previous example \eqref{eq:SIM} is known, then a natural method to estimate the unknown family $(F_u)_u$ of stochastically ordered CDFs is IDR as introduced by \citet{Henzi2019}, see also \citet{Moesching2020}. IDR is a nonparametric technique to estimate conditional distributions under stochastic ordering constraints. In brief, IDR works as follows. Given training data $(\vartheta_1,y_1),\dots,(\vartheta_n,y_n)$, where $\vartheta_i \in \Theta$ for some partially ordered set $\Theta$, IDR yields the unique optimal vector $\mathbf{\hat{F}} = (\hat{F}_1,\dots,\hat{F}_n)$ of CDFs that minimizes
\[
\frac{1}{n}\sum_{i=1}^n \crps(F_i,y_i),
\]
over all vectors $(F_1,\dots,F_n)$ of CDFs that respect the stochastic ordering constraints $F_i \st F_{j}$ if $\vartheta_i \preceq\vartheta_j$, $i, j=1,\dots,n$. Here, for any CDF $F$ and $y \in \mathbb{R}$,
\begin{equation} \label{eq:crps}
\crps(F,y) = \int_{\mathbb{R}} \left( F(z) - \one \{ y \le z \} \right)^2 \diff z 
\end{equation}
is the widely applied proper scoring rule called the continuous ranked probability score \citep[CRPS,][]{Matheson1976,Gneiting2007}.
If we have a sample $(x_1,y_1),$ \dots, $(x_n,y_n)$ from $(X,Y) \in \mathbb{R}^d \times \mathbb{R}$, we can apply IDR to the training data $(\alpha_0^\top x_1,y_1)$, \dots, $(\alpha_0^\top x_n,y_n)$, that is, we set $\vartheta_i = \alpha_0^\top x_i$, $i=1,\dots,n$ and $\Theta=\mathbb{R}$. This yields a distributional regression model for $(X,Y)$ that may be used to provide probabilistic predictions for $Y$ given $X$, see {\orange \citet[Section 2.5]{Henzi2019}} and Section \ref{sec:prediction}. 

DIMs are closely related to generalized linear models, which assume that the conditional distributions $(F_{u})_u$ belong to a known exponential family of distributions with mean $\mathbb{E}(Y\mid X = x) = g(\alpha_0^Tx)$, where $g$ is a fixed, strictly monotone link function. In fact, the Gaussian, Poisson, Gamma and Binomial GLM can be subsumed under the DIM, since they also satisfy the stochastic ordering constraint on the conditional distributions. Our approach, to leave the conditional distributions $(F_{u})_u$ unspecified, is already widely applied in classical regression for the mean, where models of the type $\mathbb{E}(Y\mid X = x) = g(\alpha_0^Tx)$ with unknown link function $g$ are called single index models. Typically, $g$ is assumed to be a smooth function and estimated by kernel regression or local polynomial approximation \citep{Hardle1993} or local polynomial approximation \citep{Carroll1997, Zou2014}. More recently, shape constrained single index models have been considered with monotone \citep{Balabdaoui2019} and convex \citep{Kuchibhotla2017} link functions. DIMs directly extend monotone single index models for the mean, since the stochastic ordering assumption on the conditional distributions implies an isotonic conditional mean function.

There is a vast literature on the estimation of the index in single index models, and we refer to \cite{Lanteri2020} for a comprehensive overview. In Section \ref{sec:estimation}, we discuss estimators for the index and the distribution functions in DIMs. Briefly, when IDR is used to estimate the conditional distribution functions, then it is sufficient to know the index function up to isotonic transformations, i.e.~to find a pseudo index function that approximates the \emph{ordering} implied by the true index. This approach is supported by the asymptotic analysis in Section \ref{sec:consistency}, which shows that when a monotone transformation of the estimated index function is consistent at the parametric rate, then a DIM with that index estimator is consistent.

A major application of distributional regression techniques is forecasting. It has been recognized in many problems, such as weather prediction or economic forecasting, that point forecasts are unable to account for the full forecast uncertainty and should be replaced by probabilistic forecasts \citep{Gneiting2014}. Distributional regression methods are statistical tools to provide such probabilistic forecasts. One fundamental contribution of DIMs is that they allow to associate a natural distributional prediction to point forecasts: If a point forecast from a statistical model is taken as the index in a DIM, for example the estimated conditional expected value, then the DIM naturally extends this deterministic forecast to a probabilistic one. Moreover, the only prerequisite is an isotonic relationship between the point forecast and the outcome in a stochastic ordering sense, which is often a natural and intuitive assumption for reasonable point forecasts.

In Section \ref{sec:data}, we use a DIM for predictions in a highly challenging dataset on the length of stay (LoS) of intensive care unit (ICU) patients. Accurate LoS predictions could serve as a tool for ICU physicians, for example to plan the number of available beds, or to identify potential long stay patients at an early stage. Moreover, the same models that are used for prediction may also be used for risk-adjustment and benchmarking across different ICUs. In the last twenty years, there have been many approaches to find appropriate regression models for LoS, see \cite{Zimmerman2006, Moran2012, Verburg2014} for some examples and \cite{Verburg2014, Kramer2017} for literature reviews. The extant methods typically model the conditional mean and are unsatisfactory when applied for single patient predictions, since the distribution of LoS is strongly right-skewed with a large variance even after conditioning on covariates. We therefore argue that LoS predictions should be probabilistic. 
In Section \ref{sec:data}, we derive calibrated and informative probabilistic forecasts for LoS, and show that the DIM outperforms existing distributional regression methods in terms of predictive accuracy.

\section{Distributional index models}
In this section, we define the DIM in its most general form. Let $Y$ be a real-valued response, and let $X$ be covariates in some general space $\mathcal{X}$. The link between $X$ and $Y$ is the index function $\theta\colon\mathcal{X} \to \mathbb{R}^d$, where $\mathbb{R}^d$ is equipped with some partial order $\preceq$. Let further $(F_u)_{u \in \mathbb{R}^d}$ be a family of CDFs such that $F_u \st F_v$ if $u \preceq v$. The DIM then assumes that
\begin{equation} \label{eq:mainAssumption}
	\mathbb{P}(Y \leq y \mid X) = F_{\theta(X)}(y).	
\end{equation}
Due to the stochastic ordering assumption, it directly follows that the conditional distributions are ordered in the index, that is, $\theta(x) \preceq \theta(x')$ implies $F_{\theta(x)} \st F_{\theta(x')}$.

We assume further that the function $\theta$ belongs to a finite dimensional vector space $\mathcal{F}$, i.e.~a parametric model for $\theta$. If $\theta_1, \ldots, \theta_p$ are a basis of $\mathcal{F}$ and if $d = 1$, then we recover the form $\mathbb{P}(Y \leq y \mid \tilde{X} = \tilde{x}) = F_{\alpha_0^T\tilde{x}}(y)$, where $\tilde{x} = (\theta_1(x) \ldots, \theta_p(x))$, and hence, the analogy to single index models. However, the estimation procedure suggested in the next section can be applied with any dimension $d$ and any partial order $\preceq$ on $\mathbb{R}^d$.

\section{Estimation} \label{sec:estimation}
Having motivated and formalized the DIM, we propose a method for estimation. Assume that a training dataset $(x_i, y_i)$, $i = 1, \ldots, n$, of independent realizations of $(X,Y)$ satisfying the model assumption (\ref{eq:mainAssumption}) is available.

In principle, it would be desirable to have a simultaneous estimator for both the index and the distribution functions. In Section \ref{sec:consistency}, we show that simultaneous estimation is possible theoretically, but computationally infeasible. The method we propose here, and for which we provide asymptotic results, is a two-stage estimation in which first the index $\theta$ is estimated, say by $\hat{\theta}$, and then the conditional CDFs based on pairs $(\hat{\theta}(x_i), y_i)$. This is inspired by the 'plug-in estimators' for monotone single index models suggested in \cite{Balabdaoui2019}.
{\orange
The estimation procedure is straightforward and reads as follows:
\begin{enumerate}
\item Estimate $\theta$ with some estimator $\hat{\theta}$ on the data $(x_i, y_i)_{i = 1}^n$,
\item compute the in-sample predictions $\vartheta_i = \hat{\theta}(x_i)$, $i = 1, \ldots, n$,
\item estimate the distribution functions $\hat{F}_{u}, u \in \mathbb{R}^d$, using $(\vartheta_i, y_i)_{i = 1}^n$.
\end{enumerate}

In the next two subsections, we reverse the order of the estimation procedure and first suggest our method for Step 3, because this has important implications for the choice of the index estimators in Step 1.
}
\subsection{Isotonic distributional regression} \label{subsec:idr}
Because of model assumption (\ref{eq:mainAssumption}), we seek an estimator $\hat{F}_u, u \in \mathbb{R}^d,$ such that $\hat{F}_u \st \hat{F}_v$ if $u \preceq v$, i.e.~$\hat{F}_u(y) \geq \hat{F}_v(y)$ for all $y \in \mathbb{R}$ and given $u, v$.
For fixed $y$, this suggests to define $\bhF = (\hat{F}_{\vartheta_{1}}, \ldots, \hat{F}_{\vartheta_{n}})$ as
\begin{equation} \label{eq:argminCdf}
 \bhF(y) \ = \argmin_{\eta_k \geq \eta_l \text{ if } \vartheta_k \preceq \vartheta_l} \sum_{i = 1}^{n} (\eta_i - \one\{y_i \leq y\})^2.	
\end{equation}
It turns out that (\ref{eq:argminCdf}) indeed yields a collection of well-defined conditional CDFs, and this estimator is called the IDR in \cite{Henzi2019}. By \citet[Theorem 2.2]{Henzi2019}, IDR can equivalently be defined in terms of conditional quantile functions, $\bhq = (\hat{q}_{\vartheta_{1}}, \ldots, \hat{q}_{\vartheta_{n}})$, where
\begin{equation} \label{eq:argminQuantiles}
 \bhq(\alpha) \ = \argmin_{\beta_k \leq \beta_l \text{ if } \vartheta_k \preceq \vartheta_l} \sum_{i = 1}^{n} (\one\{y_i \leq \beta_i\} - \alpha)(\beta_i - y_i)
\end{equation}
for any $\alpha \in (0,1)$, and the $\argmin$ is defined as the componentwise smallest minimizer if it is not unique. IDR estimates the conditional distributions non-parametrically under the stochastic order constraints. For IDR, the index $u$ can take values in any partially ordered set $\Theta$. The particular choice of the loss functions, i.e.~the squared error for the estimation of probabilities in (\ref{eq:argminCdf}) and the classical quantile loss function 
in (\ref{eq:argminQuantiles}), is in fact irrelevant here: Any other consistent loss function for the expectation or quantiles would yield the same result \citep{Henzi2019, Jordan2019}.

The above estimators are defined when the index $u$ (in $\hat{F}_u$ or $\hat{q}_u$) is in $\{\vartheta_{1}, \ldots, \vartheta_{n}\}\subseteq \Theta$. The CDFs or quantile functions for an arbitrary $u$ can be derived by interpolation of $\hat{F}_{\vartheta_{1}}, \ldots, \hat{F}_{\vartheta_{n}}$ or $\hat{q}_{\vartheta_{1}}, \ldots, \hat{q}_{\vartheta_{n}}$ for $\Theta=\mathbb{R}$, and a suitable generalization thereof for general partially ordered $\Theta$ \citep[Section 2.5]{Henzi2019}.

The following proposition is a direct consequence of the above formulas. It shows invariance properties of IDR, which make it a suitable method for estimating the conditional distributions in DIMs. We use the notation $\hat{F}_u(y; \, \bvartheta, \by)$ and $\hat{q}_u(\alpha; \, \bvartheta, \by)$ for the IDR CDFs and quantile functions estimated with training data $\bvartheta = (\vartheta_k)_{k = 1}^m$ and $\by = (y_k)_{k = 1}^m$.

\begin{prop}[Invariance of IDR]  \label{prop:invariance}
Let $\by = (y_k)_{k = 1}^m \in \mathbb{R}^m$ and $\bvartheta = (\vartheta_k)_{k = 1}^m \in \Theta^m$, and let $\Theta'$ be a partially ordered set with order $\preceq'$. Let further $g: \Theta \rightarrow \Theta'$ be such that $\vartheta_k \preceq \vartheta_l$ if and only if $g(\vartheta_k) \preceq' g(\vartheta_l)$ and $h: \mathbb{R} \rightarrow \mathbb{R}$ be strictly increasing. Define $g(\bvartheta) = (g(\vartheta_k))_{k = 1}^m$. Then, for $j = 1, \ldots, m$, $y \in \mathbb{R}$, $\alpha \in (0,1)$,
\[
	 \hat{q}_{g(\vartheta_j)}(\alpha; \, g(\bvartheta), h(\by)) = h(\hat{q}_{\vartheta_j}(\alpha; \bvartheta, \by)), \quad
	 \hat{F}_{g(\vartheta_j)}(h(y); \, g(\bvartheta), h(\by)) = \hat{F}_{\vartheta_j}(y; \bvartheta, \by).
\]
\end{prop}

Proposition \ref{prop:invariance} shows that when IDR is used to estimate the conditional distributions in Step 3, then it is sufficient to know the index $\theta$ up to increasing transformations. Moreover, any isotonic transformation can be applied to the response $Y$ to simplify the estimation of $\theta$ in Step 1, and then reverted by its inverse, without affecting the estimation of the conditional distributions. Hence, the task of estimating the index function $\theta$ is simplified to finding an estimator for a pseudo index that induces the same ordering on $\theta(x_i)$, $i = 1, \ldots, n$.

\subsection{Index estimators} \label{sec:index}

A simple but effective way to estimate the index in DIMs are classical generalized linear models. This might be surprising, because it seems that a parametric assumption has to be imposed on the distribution functions $(F_u)_{u}$ for this approach. However, due to the invariance of DIMs under monotone transformations (Proposition \ref{prop:invariance}),
it is sufficient that such a parametric assumption holds only approximately, in the sense that a monotone transformation of the index estimator converges to the index function; see Assumption (A4) in Section \ref{sec:consistency}. The only requirement is that the linear predictor of the GLM exhibits an isotonic relationship with the outcome. This can be verified by the rank correlation between the index and the outcome, or by plots of the empirical distribution of the outcome stratified according to the index. A further advantage of this approach is that GLMs are well-understood, implemented efficiently in nearly every statistical software, and one can directly build on extant literature from non-distributional regression to find a suitable index estimator. The effectiveness of GLMs in the context of DIMs is demonstrated in the data application in Section \ref{sec:data}.

Another powerful tool for index estimation in DIMs is quantile regression \citep{Koenker2005}. The stochastic ordering of the conditional distributions in DIMs is equivalent to the assumption that the conditional quantile functions $q_{\theta(x)}(\alpha)$ are increasing in the index $\theta(x)$ for every $\alpha \in (0, 1)$. One can thus estimate one or several quantiles by quantile regression, e.g.~the median and/or the 90\% quantile, and obtain estimates of the complete distribution by taking this (these) quantile(s) as the index (vector) in a DIM. Compared to the direct application of quantile regression for the estimation of conditional distributions, one does not need to specify a grid of quantiles over the whole unit interval and correct quantile crossings, but can focus on the estimation of a small number of quantiles that reveal the ordering of the conditional distributions.

In the case of a distributional single index model $F_{\theta(X)}(y) = F_{\alpha_0^TX}(y)$, that is a DIM with $d=1$, one might estimate the index $\alpha_0$ via methods for single index models. {\orange For the monotone single index model, efficient estimators have been developed recently \citep{Balabdaoui2019a,Balabdaoui2020}.} Index estimators for the single index model, such the one proposed in \cite{Lanteri2020}, also allow for non-monotone relationships between the index function $\alpha_0^Tx$ and the response, and hence monotonicity should be checked carefully. Compared to GLMs as a pseudo index, single index models gain flexibility by not assuming any fixed functional form of the relationship between $\alpha_0^TX$ and the outcome $Y$. The drawbacks are that it is more difficult to accommodate high dimensional categorical variables and to let numeric covariates enter the index-function in a non-linear fashion, e.g.~via polynomial or spline expansions, which is essential in our data application on ICU LoS. Since the DIM is already invariant under monotone transformations of the index function, it is questionable whether the benefits of using single index methods surpass these drawbacks. The same concerns are also valid for estimation methods for distributional single index models in the spirit of \cite{Hall2005}, which requires a notion of distance on the covariate space and is hence not directly applicable when categorical covariates are present.
{\orange
\subsection{Extension: Sample splitting and bagging} \label{sec:bagging}
The estimation procedure suggested so far uses in-sample predictions with the estimated index function, $\hat{\theta}(x_i)$, as covariates for distributional regression with IDR.   Depending on the index estimator, this strategy may be prone to overfitting. As a remedy, we propose a procedure in the spirit of (sub)sample aggregation (bagging).

Instead of estimating both the index function and the conditional distributions on the whole dataset, one may split the data (randomly) into two separate parts for these tasks, say $D_1 = \{1, \dots, \lfloor n\xi \rfloor\}$ and $D_2 = \{\lfloor n\xi\rfloor + 1, \dots, n\}$ for some $\xi \in (0,1)$. The index function is estimated with $(x_i, y_i)$, $i \in D_1$, and the second part of the data with the \emph{out-of-sample} predictions $\hat{\theta}(x_j)$, $j \in D_2$, serves as training data for IDR. To avoid that the estimated distribution functions depend on the random split of the training data, this procedure should be repeated several times, every time with a different split of the training data, and the conditional distribution functions are averaged in the end. The application of (sub-)sample aggregating ((sub-)bagging) has already been suggested in \cite{Henzi2019} in conjunction with IDR, where it yields smoother distribution functions and (in the case of subagging) reduces the computation time for larger datasets with multivariate covariates ($d \geq 2$). These advantages can also be expected for the DIM. In addition, the consistency result (Theorem \ref{thm:consistency}) still holds under sample splitting when the data is split into $D_1$ and $D_2$ at a constant fraction $\xi \in (0,1)$.
}
\section{Prediction} \label{sec:prediction}
This section reviews basic tools for the evaluation of probabilistic forecasts, and related properties of DIMs when used for forecasting. We denote by $F$ a generic, random probabilistic forecast for a random variable $Y$, and all probability statements are understood with respect to the joint distribution of $F$ and $Y$, which we denote by $\mathbb{P}$. For the distributional index model, the randomness of $F = F_{\theta(X)}$ is fully captured in the index $\theta(X)$.

As argued in \cite{Gneiting2007}, \emph{calibration} is a minimal requirement for probabilistic forecasts, meaning that the forecast should be statistically compatible with the distribution of the response. Of particular interest for DIMs is threshold calibration, requiring 
\begin{equation} \label{eq:thrCalib}
	\mathbb{P}(Y \leq y \mid F(y)) = F(y) , \quad y \in \mathbb{R}.
\end{equation}
It is shown in \cite{Henzi2019} that IDR, and hence also the DIM,
is always in-sample threshold calibrated, that is, (\ref{eq:thrCalib}) holds when $\mathbb{P}$ is the empirical distribution of the training data used to estimate the distribution functions. Threshold calibration can be assessed by reliability diagrams \citep{Wilks2011}, in which estimated forecast probabilities $\hat{F}(y)$ are binned and compared to the observed event frequencies in each bin. Another prominent tool for calibration checks is the probability integral transform (PIT)
\begin{equation}  \label{eq:PIT} 
	Z = F(Y-) + V \left( F(Y) - F(Y-) \right), 
\end{equation}
where $V$ is uniformly distributed on $[0,1]$ and independent of $F$ and $Y$, and $F(y-) = \lim_{z \uparrow y} F(z)$. If $Z$ is uniformly distributed, then the forecast $F$ is said to be probabilistically calibrated. The PIT can be used to identify forecast biases as well as underdispersion and overdispersion \citep{DieboldGuntherETAL1998,Gneiting2007}.

Among different calibrated probabilistic forecasts, the most informative forecast is arguably the one with the narrowest prediction intervals. This property, which only concerns the forecast distribution $F$, is referred to as \emph{sharpness} \citep{Gneiting2007}. Sharpness and calibration are often assessed jointly by means of proper scoring rules \citep{Gneiting2007a}, which map probabilistic forecasts and observations to a numerical score. An important example is the CRPS defined at (\ref{eq:crps}). IDR enjoys in-sample optimality among all stochastically ordered forecasts with respect to a broad class of proper scoring rules, including the CRPS and weighted versions of it, that is,
\begin{equation*}
\crps_{\mu}(F,y) = \int_{\mathbb{R}} \left( F(z) - \one \{ y \le z \} \right)^2 \diff \mu(z), 
\end{equation*}
where $\mu$ is a locally-finite measure. This emphasizes that IDR is a natural way to estimate the probability distributions in DIMs, since it is not tailored to a specific loss function.

\section{Consistency} \label{sec:consistency}

\subsection{Two stage estimation}
We work with a triangular array of random elements $(X_{ni}, Y_{ni}) \in \mathcal{X} \times \mathbb{R}$, $i = 1, \ldots, n$, and assume that for all $n$, the following hold:
\begin{itemize}
\item[(A1)] The random elements $X_{ni}$, $i = 1, \ldots, n$, are independent and identically distributed, and $Y_{ni}$, $i = 1, \ldots, n$, are independent conditional on $(X_{ni})_{i = 1}^n$ with
\[
	\mathbb{P}(Y_{ni} \leq y \mid X_{ni}) = F_{\theta(X_{ni})}(y),
\]
where $\theta: \mathcal{X} \rightarrow \mathbb{R}$ is a function and $(F_u)_{u \in \mathbb{R}}$ is a family of distributions such that $F_u \st F_v$ if $u \leq v$.
\item[(A2)] There exists a constant $L > 0$ such that for all $u, v, y \in \mathbb{R}$,
\[
	|F_u(y) - F_v(y)| \leq L|u-v|.
\]
\item[(A3)] On an interval $I$, the random variables $\theta(X_{ni})$ admit a density with respect to the Lebesgue measure which is bounded from below by $C_1 > 0$ and from above by $C_2$. 
\item[(A4)] There exist a strictly increasing function $g: \mathbb{R} \rightarrow \mathbb{R}$ and a constant $C_0 > 0$ such that
\[
	\lim_{n \rightarrow \infty} \mathbb{P}\left(\sup_{x \in \mathcal{X}}|g(\hat{\theta}_n(x)) - \theta(x)| \geq C_0(\log(n)/n)^{1/2}\right) = 0.
\]
\end{itemize}
We denote by $\hat{F}_{n; u}$ the IDR estimator computed with training data $(\hat{\theta}_n(X_{nj}), Y_{nj})_{j = 1}^n$, i.e. 
\[
	\hat{F}_{n; u}(y) = \hat{F}_u(y; (\hat{\theta}_n(X_{nj}))_{j = 1}^n, (Y_{nj})_{j = 1}^n),
\]
with the notation of Section \ref{subsec:idr}.

\begin{thm}[Consistency of DIM] \label{thm:consistency}
Under assumptions (A1)-(A4), there exists a constant $C > 0$ such that
\[
	\lim_{n \rightarrow \infty} \mathbb{P}\left(\sup_{y \in \mathbb{R}, x \in \mathcal{X}_n} |\hat{F}_{n; \hat{\theta}_n(x)}(y) - F_{\theta(x)}(y)| \geq C \Big(\frac{\log n}{n}\Big)^{1/6}\right) = 0,
\]
where $\mathcal{X}_n = \{x \in \mathcal{X}: \, [\theta(x) \pm (\log n / n)^{1/6}] \subseteq I\}$.
\end{thm}

An analogous result to Theorem \ref{thm:consistency} can be shown for the variant of the DIM with sample splitting described in Section \ref{sec:bagging}. The requirements under sample splitting are slightly weaker, namely, the density of $\theta(X_{ni})$ does not have to be bounded from above in (A2), and in (A4), it is sufficient that the index estimator $\hat{\theta}_n$ converges at a rate of $o((\log(n)/n)^{1/3})$ instead of $n^{-1/2}$. The resulting convergence rate of the DIM with sample splitting is of order at least $(\log(n)/n)^{1/3}$. The proofs of Theorem \ref{thm:consistency}, both, with and without sample splitting, rely on the consistency results about the monotonic least squares estimator in \cite{Moesching2020}, and are given in Appendix \ref{app:proof}. 

Assumption (A1) is the basic model assumption of DIMs. The Lipschitz-continuity in (A2) also appears in the monotone single index model for the mean \citep{Balabdaoui2019}. Since the distributional single index model and the monotone single index model are equivalent when $Y$ is binary, the Lipschitz assumption (A2) is natural in this context; also (A3) can be derived from the assumptions in \cite{Balabdaoui2019}. Assumptions (A2) and (A3) are required for the consistency of the monotone least squares estimator, with (A3) ensuring that the 'design points' $\theta(X_{nj})$ are dense enough in a region of interest, c.f.~\cite{Moesching2020}. {\orange A parametric model $\theta = \alpha_1 \theta_1 + \dots + \alpha_p \theta_p$ satisfies this assumption when at at least one of the summands $\alpha_i\theta_i$ admits a continuous distribution on $I$ with density bounded away from zero.} In (A4), we require uniform consistency of a monotone transformation of the index estimator at a rate of $n^{-1/2}$, i.e.~not necessarily consistency of the index estimator itself. In a parametric model $\theta = \alpha_1 \theta_1 + \dots + \alpha_p \theta_p$, uniform consistency is satisfied for any $\sqrt{n}$-consistent estimator of the coefficients $\alpha_1, \ldots, \alpha_p$, when the functions $\theta_1, \ldots, \theta_p$ are bounded. All estimators suggested in Section \ref{sec:index} are consistent at the rate $n^{-1/2}$ under suitable conditions.

\subsection{Simultaneous estimation}
In this subsection, we treat the question to what extent simultaneous estimation of the index and the distribution functions is possible and sensible in the DIM. Currently, the results are of theoretical interest only. 

It has been shown in \cite{Balabdaoui2019} that for the monotone single index model, there exists a simultaneous minimizer $(\psi_0, \alpha_0)$ of the squared error
\begin{equation*}
	\sum_{i = 1}^n (\psi_0(\alpha_0^Tx_i) - y_i)^2
\end{equation*}
where $\psi_0: \mathbb{R} \rightarrow \mathbb{R}$ is an increasing function, $\alpha_0 \in  \{x \in \mathbb{R}^p: \|x\| = 1\}$ is the index, and $(x_i, y_i) \in \mathbb{R}^p \times \mathbb{R}$, $i = 1, \ldots, n$. The minimizer is in general not unique. 

A similar result also holds in the distributional index model, when the loss function is defined as 
\begin{equation} \label{eq:simult_loss}
	l(\hat{\theta}, \bhF) = \sum_{i = 1}^n \crps(\hat{F}_{\hat{\theta}(x_i)}, y_i).
\end{equation}
For basis functions $\theta_1, \ldots, \theta_p$ of the vector space $\mathcal{F}$ containing the true index function $\theta$, every index estimator $\hat{\theta}: \mathcal{X} \rightarrow \mathbb{R}^d$ can be written as $\hat{\theta} = \hat{\alpha}_1 \theta_1 + \dots + \hat{\alpha_p}\theta_p$. The loss (\ref{eq:simult_loss}) has a unique minimizer $\bhF = (\hat{F}_{\hat{\theta}(x_i)}, \ldots, \hat{F}_{\hat{\theta}(x_n)})$ for fixed $\hat{\theta}$, namely the IDR. This minimizer only depends on $\hat{\theta}$ via the partial order on the points $\hat{\theta}(x_i)$, $i = 1, \ldots, n$. But the number of partial orders on $n$ points is finite, and so there exists a minimizer of (\ref{eq:simult_loss}).

In general, the number of partial orders induced by index functions $\hat{\theta}$ is too large for a direct minimization of (\ref{eq:simult_loss}) to be possible: When $\mathcal{X} = \mathbb{R}^p$ and $\theta_1, \ldots, \theta_p$ are the coordinate projections, then the number of total orders grows at a rate of $n^{2(p-1)}$ \citep{Balabdaoui2019}. Moreover, when the index space is partially but not totally ordered, trivial solutions (a perfect fit to the training data) may appear, namely if the points $\hat{\theta}(x_i)$, $i = 1, \ldots, n$, are all incomparable in the partial order. Hence, the simultaneous estimation of the index and the distribution functions in DIMs is generally not feasible. A related interesting question for further research is to find a way to directly parametrize and estimate partial orders for isotonic or isotonic distributional regression, instead of indirectly via an index function.

\section{Data application} \label{sec:data}
We apply a DIM to derive probabilistic forecasts for intensive care unit (ICU) length of stay (LoS) based on patient information available 24 hours after admission. The main difficulty of such predictions is that, even conditional on many demographic and physiologic patient specific covariates, there is often great uncertainty in the LoS. In addition to unknown factors (e.g.~frailty status, patient or family wishes), the LoS also depends on non-patient-related information such as ICU organization and resources. We therefore model the LoS using data of single ICUs rather than a merged dataset, thus keeping the ICU-related variables fixed. This allows forecasts within each single ICU as well as the comparison of the forecasted LoS of patients across ICUs. The same methodology can also be used on a joint dataset of several ICUs, giving a reference LoS forecast on the combined case-mix. Using these predictions for risk-adjustment and benchmarking is promising but goes beyond the scope of this paper.

All computations in this application were performed in R 4.0 \citep{R} using the packages \verb!mgcv! \citep{Wood2017} for the estimation of index models and Cox proportional hazards regression, \verb!quantreg! \citep{Koenker2020} for quantile regression, and \verb!isodistrreg! \citep[\url{https://github.com/AlexanderHenzi/isodistrreg}]{Henzi2019} for IDR.

\subsection{Data and variables}

\begin{sidewaystable}
\caption{Covariates used for ICU length of stay predictions. Availability of the variables is given by 'admission' (at patient admission) or by the number of hours after admission.} \label{tab:vars}
\bigskip
\begin{tabularx}{\textwidth}{l|l|X}
\toprule
\textbf{Variable} & \textbf{Availability} & \textbf{Description}\\
\hline\hline
Age & admission & patient age at admission \\
\hline
Sex & admission & male, female \\
\hline
Planned & admission & admission is announced at least 12h in advance (true/false) \\
\hline
Readmission & admission & patient was discharged from the same ICU at most 48 hours ago (true/false)\\
\hline
Admission source & admission & admission source (emergency room; intermediate care unit, high dependency unit, recovery room; hospital ward; surgery; others) \\
\hline
Location before & admission & location before \emph{hospital} admission (home; other hospital; others)\\
hospital admission & & \\
\hline
Diagnosis & 24h & main diagnosis on first day (structured into: cardiovascular, respiratory, gastrointestinal, neurological, metabolic, trauma, others; in total 36 different specific ICU relevant diagnoses)\\
\hline
NEMS & 8h & NEMS \citep{Miranda1997} over first shift after patient admission (8-12h)\\
\hline
SAPS II & 24h & \cite{le1993}\\
\hline
Interventions & 24h & interventions 24 hours before until 24 hours after admission (13 categories of interventions, e.g.~surgeries, interventions in respiratory system, cardiovascular interventions)\\
\bottomrule
\end{tabularx}
\end{sidewaystable}

Since 2005, the Swiss Society of Intensive Care Medicine collects ICU key figures and information on patient admissions in the Minimal Dataset of the Swiss Society of Intensive Care Medicine (MSDi). Our analysis is based on a part of this dataset suitable for LoS predictions, namely, we include 18 out of 86 ICUs which, after the application of selection criteria described below, include more than 10'000 patient admissions. The codes used as identifiers for the ICUs were generated randomly. The sample sizes range from 10'041 to 36'865 with an average of 17'181 observations per ICU. The cutoff of 10'000 is based on our experience with IDR and probabilistic forecasts in general, which require sufficiently large datasets for a meaningful and stable evaluation, {\orange especially when the models involve large numbers of covariates and a skewed response variable, as it is the case here}. However, the prediction methods can also be applied to smaller datasets.

Based on literature review, we identified the variables described in Table \ref{tab:vars} as relevant for LoS forecasts (\cite{Zimmerman2006}; \citet[Table S1]{Verburg2014}; \cite{Niskanen2009}). We exclude patients that were transferred from or to another ICU, because their LoS is incomplete. As in \cite{Zimmerman2006}, we also remove patients younger than 16 years and patients admitted after transplant operations or because of burns. 
Patients with missing values in the variables in Table \ref{tab:vars} are excluded, too.

Table \ref{tab:vars} documents at what time after admission the relevant covariates for LoS predictions are available. While all variables are available 24 hours after patient admission, the information is completed also for patients staying at the ICU less than one day. For example, ICU interventions within the first 24 hours are then only interventions performed until patient discharge, and the SAPS II is computed based on the worst physiological values until discharge instead of the worst values in the first 24 hours at the ICU. 

\begin{figure}
\caption{Empirical distribution functions of the standardized and non-standardized LoS for selected ICUs. The standardized LoS is defined as $Y - 1 + h / 24$, where $h$ is the admission hour of a patient and $Y$ is the non-standardized LoS, i.e.~the time between patient admission and discharge. Only patients with positive standardized LoS are included. \label{fig:losStandardization}}
\bigskip
\center
\includegraphics[width = 0.7\textwidth]{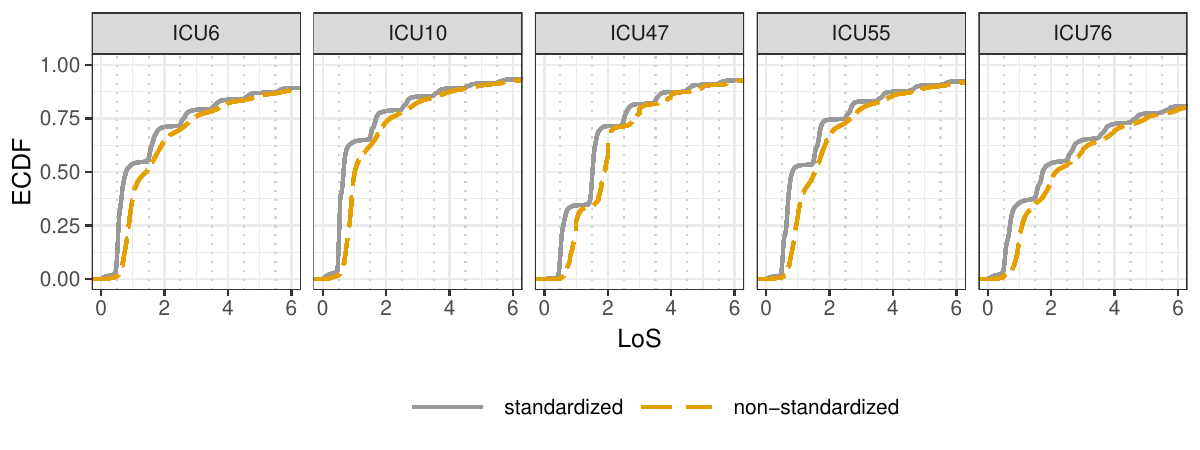}
\end{figure}

In preliminary tests, we found that for probabilistic LoS forecasts, the usual definition of LoS as the time between patient admission and discharge is problematic, because most ICUs discharge patients during specific time windows, but the admission times are spread throughout the day. As a consequence, it may happen that the predicted LoS for certain patients does not conform with the discharge practice of a ICU, e.g.~there might be a high predicted probability for a patient being discharged around midnight but the ICU actually discharges patients in the early afternoon. To circumvent this problem, we decided to measure the LoS as the \emph{time between the next midnight after patient admission until discharge}, thereby standardizing all admission to the same (day)time and revealing the true pattern in the patient discharge times; see Figure \ref{fig:losStandardization}. All results in this section use this definition of the LoS. Patients who do not stay over at least two calendar days are excluded, which is unproblematic since in practice, the data required for predictions is only available 24 hours after admission and the forecast should be conditioned on the event that the patient already stayed at the ICU for 24 hours. Forecasts for the non-standardized LoS, i.e.~the time between admission and discharge, can be derived via the relation
\[
	\mathbb{P}(Y > 1 + t | Y > 1) = \frac{\mathbb{P}(\tilde{Y} > t + h / 24 | \tilde{Y} > 0)}{\mathbb{P}(\tilde{Y} > h / 24 | \tilde{Y} > 0)},
\]
where $Y$ and $\tilde{Y} = Y - 1 + h/24$ denote the LoS and the standardized LoS measured in days, respectively, and $h$ the admission hour of a given patient. Since only patients staying at least until midnight of the admission day are used as training data, our LoS forecasts are conditioned on the event $\{\tilde{Y} > 0\}$ in the above equation.

We select the most recent 20\% of the observations in each ICU for model validation, thereby mimicking a realistic situation in which past data are used to predict the LoS of present and future patients. This implies that forecasts might be inaccurate if the relationship between the covariates and LoS changes over time, and it is part of our analysis to check to what extent past data can be reasonably used to predict the LoS of future patients. Of the remaining data, randomly selected 75\% are used for model fitting and 25\% for model selection via out-of-sample predictions.  All comparisons of different variants of a distributional regression model were performed by such out-of-sample predictions.

\subsection{Derivation of DIM}
To derive an index estimator for the DIM, we can benefit from the comparisons of regression models for point forecasts for LoS in the extant literature. \cite{Moran2012} and \cite{Verburg2014} found that a Gaussian linear regression for the expected log-LoS is suitable for point forecasts, and we use this as our candidate for the index estimator and will refer to it as the 'lognormal index model'. We use the transformation $y \mapsto \log(y + 1)$, which results in more symmetric distributions than the logarithm. All variables from Table \ref{tab:vars} were included in the model, and the effects of the continuous variables age, SAPS and NEMS were modeled by cubic regression splines. Interactions of variables were explored but not included in the final model. We also tested whether merging factor levels with few observations improved the model, but the untransformed covariates yielded the best forecasts in out-of-sample predictions on the part of the data used for model selection.

We tested two other index estimators for the expected LoS to investigate the robustness of the DIM with respect to the index. The first one estimates the expected log-LoS under the assumption of a scaled t-distribution. The mean is modeled as a function of the covariates, with the same specification as for the lognormal index model, and the degrees of freedom are estimated, with a minimal threshold of $5$ to ensure stability. This model is structurally similar to the lognormal index model, but more robust with respect to outliers, which occur even after the log-transformation. The second alternative is a gamma regression for the untransformed LoS with logarithm as the link function. While the three index models yield different predictions on the scale of the LoS, they largely agree when only the \emph{ordering} of the predictions is considered: Over the 18 ICUs, the rank correlation between predictions by two of the models is $0.98$ on average with a minimum of $0.86$. As a consequence, there is no significant difference between the corresponding DIM forecasts: Evaluated on the dataset for model selections, the average CRPS over all ICUs of DIM forecasts based on different models only differs by up to $0.01$, while the averages are around $1.40$. The predictions based on the lognormal index model achieved the best results in most ICUs and were therefore selected for the predictions on the validation data.

Due to the large training datasets, splitting of the training data as described in Section \ref{sec:bagging} only has a marginal effect on the predictions. Estimating the index function on the full training data and the conditional distributions on in-sample predictions only increased the average CRPS by $0.01$ (on $1.40$), compared to a bagging approach with 100 random splits of the training data into equally sized parts for the estimation of the index and the CDFs. For the final evaluation, we show the results of the simpler variant without bagging.

Figure \ref{fig:checkMonotonicity} illustrates how to perform a check of the stochastic ordering assumption of the DIM: We bin the observed LoS according to the index value, and plot the empirical cumulative distribution functions (ECDFs) of the LoS in each bin. By varying the positions and sizes of the bins, it can be seen that the empirical distributions are indeed sufficiently well ordered. The Spearman correlation between the index and the observed LoS is $0.53$ on average over all ICUs (range $0.40-0.65$), which confirms that there is an isotonic relationship between the index and the actual LoS for most ICUs, {\orange taking into account the high uncertainty in the
LoS of ICU patients even conditional on patient information collected at the first day.}

\begin{figure} 
\caption{(a) Index function and $\log(\mathrm{LoS} + 1)$ for selected ICUs. (b) ECDFs of the LoS stratified into the bins given by the vertical shaded stripes in panel (a).\label{fig:checkMonotonicity}}
\bigskip
\center
\includegraphics[width = 0.7\textwidth]{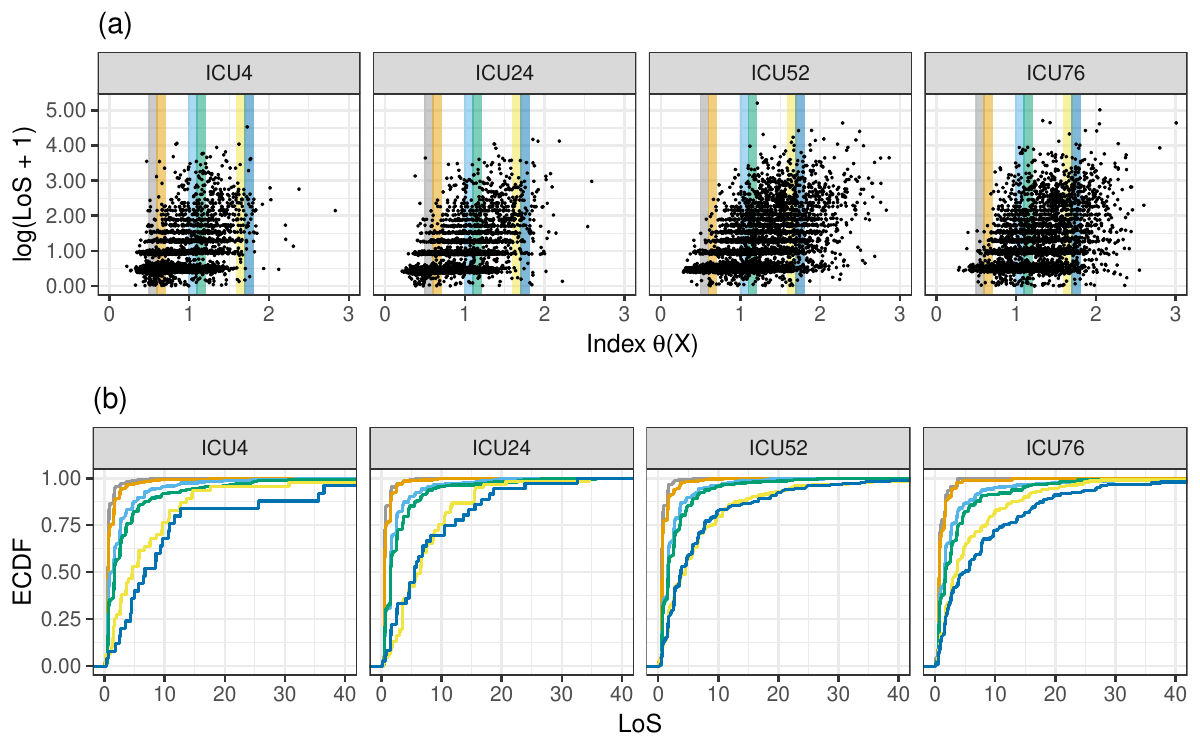}
\end{figure}

\subsection{Alternative regression methods}
We compare the DIM to two other distributional regression methods: A Cox proportional hazards model \citep{Cox1972} and quantile regression with monotone rearrangement \citep{Koenker2005, Chernozhukov2010}. For both, we use the same variables and specifications as in the index estimator for the DIM, which was superior compared to other variants tested; detailed results are provided in the appendix.

A Cox proportional hazards model is a classical choice for modeling survival times, and it shares some similarities with a DIM. Both models are semi-parametric and based on stochastic order restrictions on the conditional distributions, namely the usual stochastic ordering in the DIM and the hazard rate order in Cox regression, which is stronger than the usual stochastic order \citep[Theorem 1.B.1]{Shaked2007}. While the distribution functions are estimated non-parametrically in Cox regression, the relationship between different conditional distributions is modeled parametrically via the hazard ratio, as opposed to the DIM, where only the ordering on the conditional distributions is modeled parametrically by the index function.

Quantile regression, on the other hand, imposes less assumptions on the conditional distributions. The conditional quantiles are modeled separately and satisfy no stochastic order constraints. In particular, if there are strong violations of the stochastic order assumptions of the DIM or Cox regression, we would expect that the more flexible quantile regression achieves better forecasts by fitting crossing quantile curves for different patients. This allows an informal check of the underlying assumptions of Cox regression and the DIM (see Figure \ref{fig:checkRqCrossings}). We use a grid of quantiles from $0.005$ to $0.995$ with steps of $0.001$, which gave better results than a coarser grid with steps of $0.01$.

We also tested fully parametric models of GAMLSS type, and kernel methods as implemented in the {\tt np} package in R \citep{Hayfield2008}. Unfortunately, we could not find a sufficiently flexible parametric family for a GAMLSS, and the application of kernel methods was not feasible due to computational problems with the large datasets and high numbers of covariates. As for the DIM, computation is obviously more demanding than for fully parametric methods, but still fast thanks to the sequential implementation of IDR described in \citet{Henzi2020}. On a personal computer with Intel(R) Core i7-8650 CPU, computation with the lognormal index model without bagging takes 3 seconds for the smallest ICU (6'024 observations in training dataset) and 25 seconds for the largest ICU (22'219 observations). Estimation and prediction on the total dataset (all 18 ICUs) require about 2.5 minutes.

\begin{figure} 
\caption{Predictive CDFs for four selected patients based on the training data of the ICU the patients were admitted to. \label{fig:exampleCdfs_revised}}
\bigskip
\center
\includegraphics[width=0.7\textwidth]{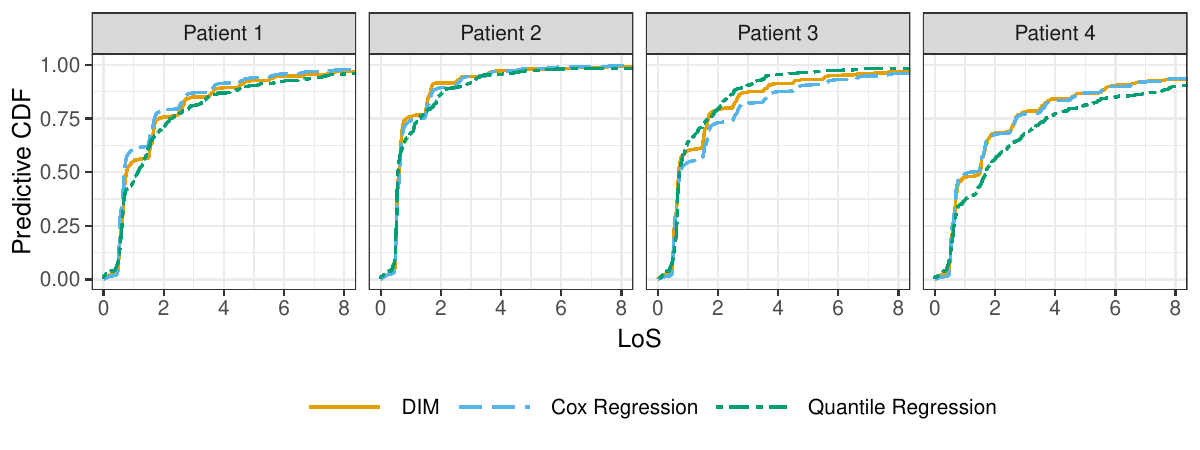}
\end{figure}

\subsection{Results}
Figure \ref{fig:exampleCdfs_revised} illustrates the probabilistic forecasts for different patients based on the training data of the ICU the patients were admitted to. Patient 1, male, 32 years old, was admitted because of a severe sepsis or septic shock. Patient 2 is a 67 years old female with aortic aneurysm or aortic dissection, Patient 3 is 58 years old, male with a metabolic decompensation, and Patient 4 is a 78 old female admitted from a high dependency unit with subarachnoidal hemorrhage. Patient 2 has the shortest predicted LoS: The DIM and Cox regression predict that she leaves the ICU at the first day after admission with a probability of almost $75\%$. For the remaining patients, the predictive CDFs are more skewed, and a LoS of more than three days is not unlikely. It is immediately visible that the DIM and Cox regression are able to recover the pattern in the ICU discharge times, with flat pieces of the CDFs around midnight. Quantile regression, on the other hand, merely interpolates this pattern.

\begin{figure}
\caption{Reliability diagrams of probabilistic forecasts for the predicted probability that the LoS exceeds $1, \, 5, \, 9, \, 13$ days. The forecast probability is grouped into the bins $[0, 0.1], (0.1, 0.2], \ldots, (0.9, 1]$ and the observed frequencies are drawn at the midpoints of the bins. Only bins with more than two observations are included. \label{fig:reliability_revised}}
\bigskip
\center
\includegraphics[width = 0.7\textwidth]{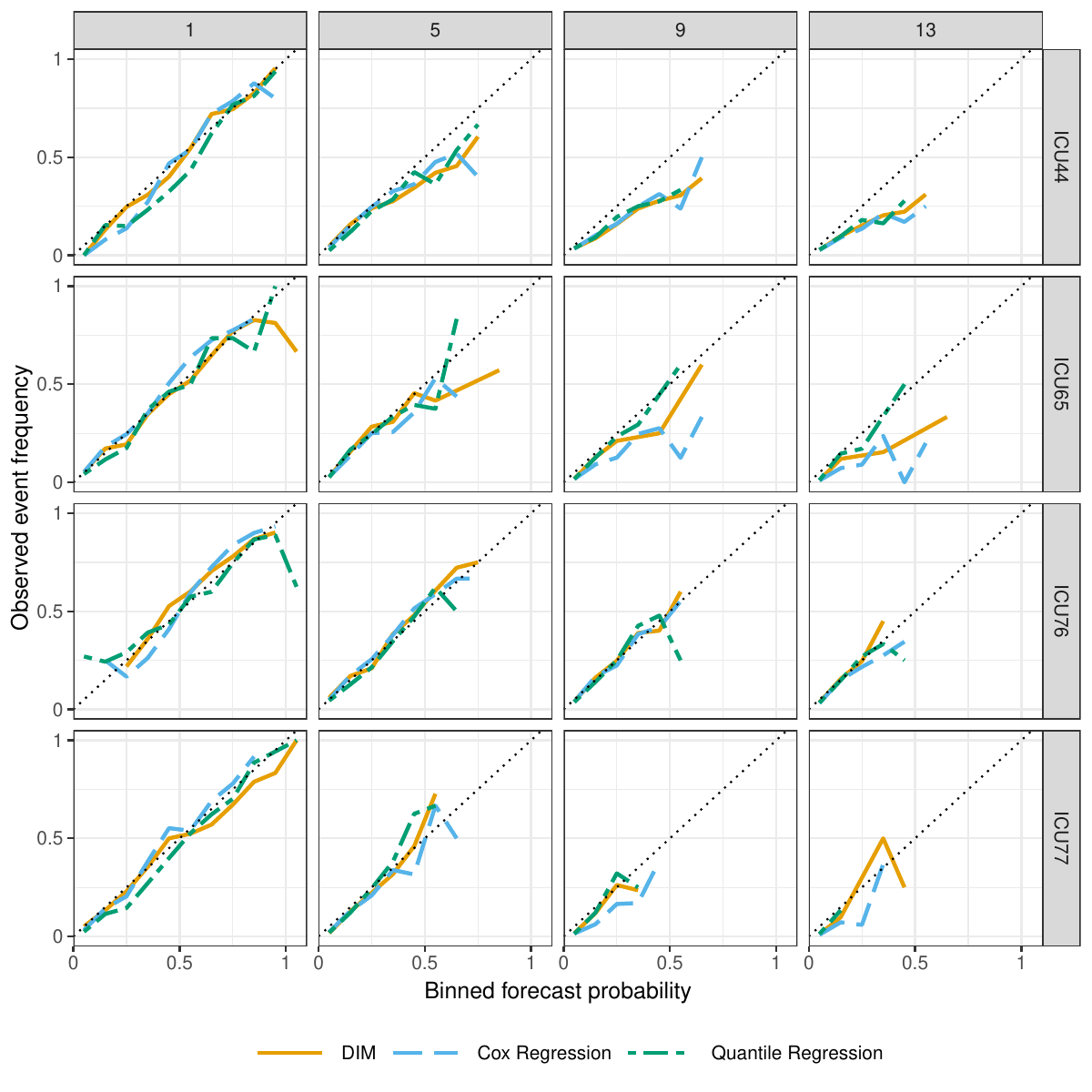}
\end{figure}

\begin{table}
\caption{Summary statistics (mean, median and standard deviation) of numeric variables in the dataset. \label{tab:summaries}} 
\bigskip
\resizebox{\linewidth}{!}{\begin{tabular}{l|r|r|r||r|r|r||r|r|r||r|r|r}
\toprule
ICU & \multicolumn{3}{c||}{LoS} & \multicolumn{3}{c||}{Age} & \multicolumn{3}{c||}{NEMS} & \multicolumn{3}{c}{SAPS}\\
\hline\hline
& mean & med. & sd & mean & med. & sd & mean & med. & sd & mean & med. & sd \\
\hline
ICU44 & 3.9 & 1.5 & 7.8 & 59.0 & 61 & 17.6 & 27.1 & 27 & 8.5 & 34.0 & 31 & 18.9\\
\hline
ICU65 & 1.8 & 0.6 & 4.3 & 67.2 & 69 & 13.9 & 25.5 & 25 & 7.9 & 28.7 & 28 & 12.5\\
\hline
ICU76 & 4.3 & 1.7 & 7.2 & 63.2 & 66 & 15.6 & 30.3 & 30 & 8.3 & 41.2 & 40 & 17.2\\
\hline
ICU77 & 1.8 & 0.6 & 3.2 & 65.0 & 68 & 15.9 & 21.9 & 18 & 8.0 & 31.1 & 28 & 16.1 \\
\bottomrule
\end{tabular}}
\end{table}

Here, detailed results are only shown for the best and worst two ICUs with respect to the CRPS of the DIM forecasts; see Appendix \ref{app:supplement} for tables and figures for all ICUs. Summary statistics of the LoS and other numeric variables for the patients of these ICUs are given in Table \ref{tab:summaries}. All probabilistic regression methods can reliably predict the probability that the LoS exceeds $k = 1, \, 5, \, 9, \, 13$ days; see Figure \ref{fig:reliability_revised}. Figure \ref{fig:pit_revised} shows that the forecasts achieve a better probabilistic calibration than the ECDF of the LoS in the training data, which is uninformative as a forecast and does not take into account changes in the ICU-case mix that are reflected in the covariates. Further improvements of calibration may be possible by selecting a tailored training dataset, taking into account organizational changes, and developments in treatments that have an influence on the LoS or on the relationship between covariates and the LoS. Such information is not available in our dataset.

\begin{figure}
\caption{PIT histograms of the probabilistic forecasts with bins of width $1/20$. \label{fig:pit_revised}}
\bigskip
\center
\includegraphics[width = 0.7\textwidth]{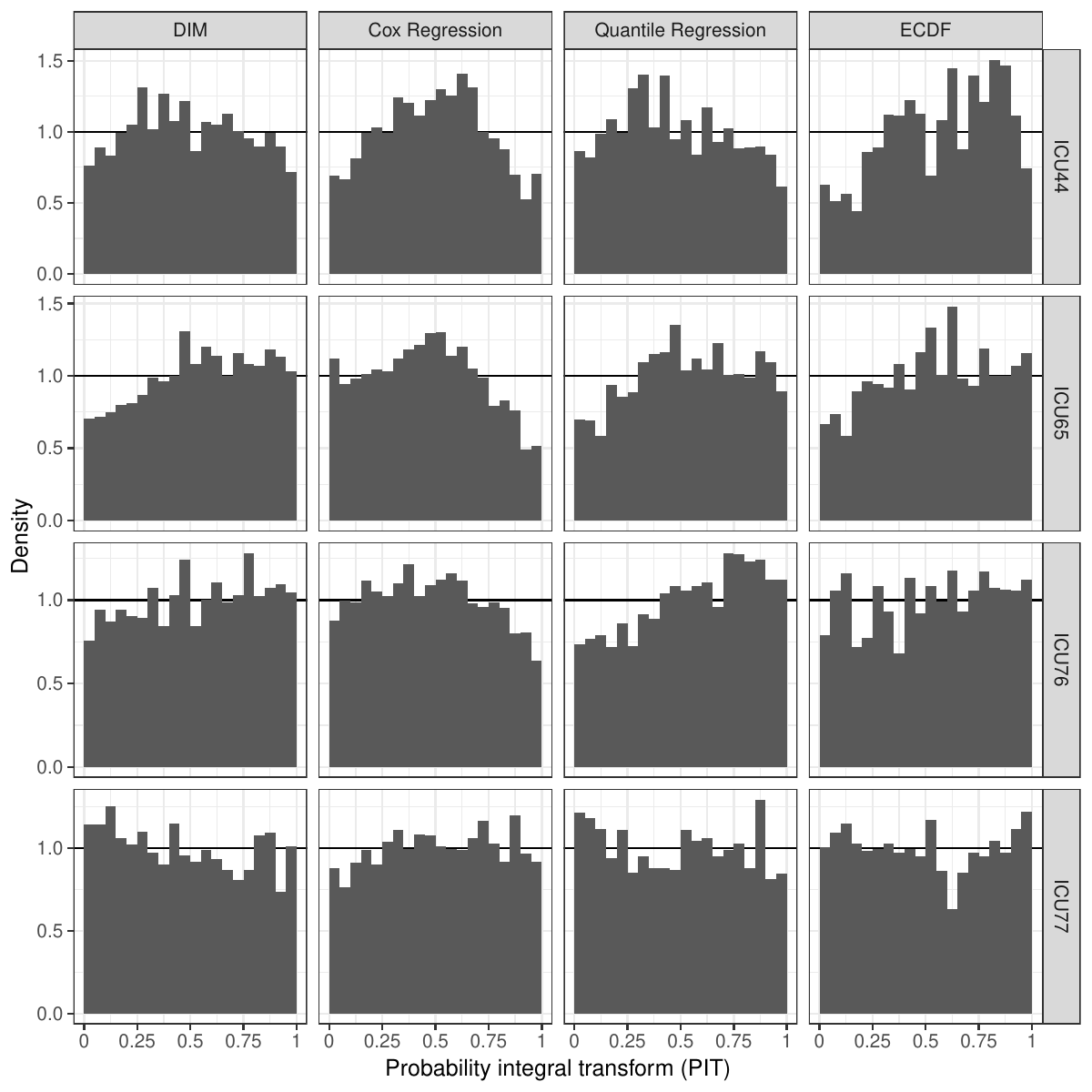}
\end{figure}

While all three distributional regression methods yield similar results in terms of calibration, there is a clear ranking with respect to forecast accuracy: In all ICUs, the DIM achieves the lowest CRPS, followed by quantile regression in second and Cox regression in third place. For comparison, Table \ref{tab:crps} also shows the CRPS of the ECDF forecast, and of the deterministic point forecast of the lognormal index model, which is its mean absolute error. Interestingly, the ECDF forecast achieves a lower mean CRPS in all ICUs (average improvement of 13\%) than the point forecast, although it does not take any covariate information into account. This highlights the superiority of even simple probabilistic forecast over point forecasts in the context of ICU LoS. A further average improvement of 13\% in the mean CRPS is achieved when going from the uninformative ECDF forecast to the worst of the probabilistic regression methods in terms of CRPS, which is Cox regression. The differences in the CRPS of the forecasts using distributional regression methods are smaller, but consistent over the ICUs: In terms of average CRPS, quantile regression outperforms Cox regression in 15 out of 18 ICUs, and the DIM outperforms Cox regression in all and quantile regression in all except 2 ICUs. The difference in CRPS between the DIM and quantile regression is highly significant when tested with Wilcoxon's signed rank test except for the ICUs with identifiers 19 and 33, where the p-values are $0.101$ and $0.219$ and quantile regression achieves lower average scores. Wilcoxon's signed rank test was applied because the CRPS differences are heavy-tailed, so a t-test is not appropriate (see Figure \ref{fig:crpsDiff_revised}).

In conclusion, with distributional regression methods and especially the DIM, it is possible to obtain reliable, reasonably well calibrated, and informative probabilistic forecasts for ICU LoS in a realistic setting. These forecasts are not only more informative than point forecasts, but also reduce the forecast error by more than 25\%.

\begin{table}
\caption{CRPS of probabilistic forecasts. The column 'Point' shows the mean absolute error of the point forecast obtained from the lognormal index model, and p-values of Wilcoxon's signed rank test for the difference in CRPS between DIM and quantile regression are given in the column labelled $p$. P-values smaller than $10^{-16}$ are written as $0$. \label{tab:crps}}
\bigskip
\center
\begin{tabular}{l|r|r|r|r|r|r}
\toprule
ICU & $p$ & DIM & Quantile reg. & Cox reg. & ECDF & Point \\
\midrule
ICU4 & $1.18\cdot 10^{-11}$ & $1.074$ & $1.076$ & $1.089$ & $1.191$ & $1.399$ \\ 
ICU6 & $3.81\cdot10^{-12}$ & $1.360$ & $1.385$ & $1.386$ & $1.605$ & $1.830$ \\ 
ICU10 & $0$ & $1.194$ & $1.221$ & $1.209$ & $1.312$ & $1.553$ \\ 
ICU19 & $1.01\cdot10^{-1}$ & $1.041$ & $1.032$ & $1.048$ & $1.189$ & $1.350$ \\ 
ICU20 & $5.13\cdot10^{-6}$ & $2.216$ & $2.223$ & $2.241$ & $2.505$ & $2.859$ \\ 
ICU24 & $0$ & $1.099$ & $1.111$ & $1.141$ & $1.265$ & $1.416$ \\ 
ICU33 & $2.19\cdot10^{-1}$ & $0.975$ & $0.974$ & $0.983$ & $1.090$ & $1.363$ \\ 
ICU39 & $1.38\cdot10^{-16}$ & $1.332$ & $1.352$ & $1.383$ & $1.697$ & $1.872$ \\ 
ICU44 & $1.06\cdot10^{-3}$ & $2.256$ & $2.259$ & $2.328$ & $2.480$ & $2.952$ \\ 
ICU47 & $3.69\cdot10^{-5}$ & $0.977$ & $0.980$ & $1.036$ & $1.231$ & $1.363$ \\ 
ICU52 & $7.40\cdot10^{-5}$ & $1.845$ & $1.866$ & $1.868$ & $2.121$ & $2.580$ \\ 
ICU55 & $0$ & $1.062$ & $1.085$ & $1.055$ & $1.253$ & $1.445$ \\ 
ICU58 & $1.25\cdot10^{-15}$ & $1.393$ & $1.409$ & $1.442$ & $1.763$ & $1.970$ \\ 
ICU65 & $0$ & $0.908$ & $0.914$ & $0.981$ & $1.062$ & $1.194$ \\ 
ICU76 & $0$ & $2.420$ & $2.448$ & $2.458$ & $2.783$ & $3.468$ \\ 
ICU77 & $1.76\cdot10^{-16}$ & $0.921$ & $0.936$ & $0.938$ & $1.117$ & $1.260$ \\ 
ICU79 & $1.86\cdot10^{-11}$ & $1.446$ & $1.457$ & $1.512$ & $2.172$ & $2.228$ \\ 
ICU80 & $0$ & $0.942$ & $0.971$ & $0.949$ & $1.094$ & $1.253$ \\ 
\midrule
Mean  & & $1.359$ & $1.372$ & $1.392$ & $1.607$ & $1.853$ \\
\bottomrule
\end{tabular}
\end{table}

\section{Discussion}

In this paper, we have introduced DIMs as intuitive and flexible models for distributional regression. Distributional regression approaches provide full conditional distributions of the outcome given covariate information, and are thus more informative than classical regression approaches for the conditional mean, median or specific quantiles. However, specifying a good distributional regression model is usually less intuitive than specifying a regression model for, say, the conditional mean. An appealing feature of DIMs is that for the modeling of the index function classical approaches and intuition for modeling a conditional mean or median can be used. Given the index function, the shape of the full conditional distribution is then learned from training data using IDR, that is, distributional regression under stochastic ordering constraints. The second step does not involve any parameter tuning or implementation choices. 

The idea of reducing the complexity of a potentially high-dimensional covariate space by using an index function in distributional regression has also been used in the work of \citet{Hall2005,Zhang2017}. In these works, the index function has to be univariate and parametrizes a distance on the covariate space that is then used for kernel methods to estimate the conditional distributions. In contrast, the index function in a DIM parametrizes partial orders on the covariate space allowing for stochastic order constrained distributional regression in the second step. 

Finding an informative index function is critical and usually requires expertise of the problem at hand. However, in many cases, existing models for the conditional mean or median can be used directly, as demonstrated in the application on ICU LoS. Indeed, it may even happen that a poorly fitting conditional mean model works well for a DIM since it is sufficient that the model is correct up to monotone transformations, or, in other words, that it is a good model for a pseudo index.

{\orange The distributional regression approach in \citet{Chernozhukov2020} allows to accomodate continuous, discrete and mixed discrete-continuous outcomes. The same is true for IDR, and thus for DIM models. While the case study in this paper concerns a continuous outcome, IDR has been successfully applied to a mixed discrete-continuous outcome in \citet{Henzi2019}. It would be interesting to investigate the different benefits and drawbacks of DIM models versus the methods of \citet{Chernozhukov2020} in particular in the case of discrete outcomes.}

Since IDR can be combined well with (sub-)bagging, the same also holds for DIMs. (Sub-)bagging is useful to avoid overfitting, may increase computational efficiency, and lead to smoother estimated conditional CDFs. We have explored bagging in our data application in Section \ref{sec:data} with relatively at hoc choices for the number of random splits of the training data. A systematic study of optimal choices for subsample sizes and/or iterations is desirable. 

A promising future extension of DIMs is to replace the IDR step by distributional regression under a stronger stochastic ordering constraint such as a likelihood ratio ordering constraint, or by a weaker one such as second order stochastic dominance. However, this requires fundamental advances concerning the estimation of distributions under these constraints.

\appendix

\section{Proof of Theorem \ref{thm:consistency}} \label{app:proof}
The following lemma is Theorem 4.6 in \citet{Moesching2020}, which we state for completeness.

\begin{lem} \label{lem:lem1}
Let $Z_1, Z_2, Z_3, \ldots$ be independent random variables with respective distribution functions $G_1, G_2, G_3, \ldots$. For $k \in \mathbb{N}$, let 
\[
	\hat{\mathbb{G}}_k(\cdot) = \frac{1}{k} \sum_{i = 1}^k \one\{Z_i \leq \cdot\} \quad \text{and} \quad
	\bar{G}_k (\cdot) = \frac{1}{k} \sum_{i = 1}^k G_i(\cdot).
\]
Then there exists a universal constant $M \leq 2^{5/2}e$ such that for all $\eta \geq 0$,
\[
	\mathbb{P}\left(\sqrt{k} \|\hat{\mathbb{G}}_k - \bar{G}_k\|_{\infty} \geq \eta \right) \leq M \exp(-2\eta^2),
\]
where $\|\cdot\|_{\infty}$ denotes the usual supremum norm of functions.
\end{lem}

The results and proofs below use the following definitions. We denote by $\lambda(J)$ the Lebesgue measure of a measurable set $J \subset \mathbb{R}$, and define the events
\begin{equation} \label{eq:Bn}
	B_n = \left\{\sup_{x \in \mathcal{X}}\, |g(\hat{\theta}_n(x)) - \theta(x)| < C_0(\log(n)/n)^{1/2}\right\}.
\end{equation}
For $1 \leq r \leq s \leq n$ and a permutation $\sigma$ of $\{1, \dots, n\}$, let
\begin{align*}
	 w_{rs} & = s - r + 1, \quad & \hat{\mathbb{F}}^{\sigma}_{rs} & = \frac{1}{w_{rs}}\sum_{i = r}^s\one\{Y_{n\sigma(i)} \leq \cdot\}, \\
	\bar{F}^{\sigma}_{\theta; rs}(\cdot) & = \frac{1}{w_{rs}} \sum_{i = r}^sF_{\theta(X_{n\sigma(i)})}(\cdot), \quad & \bar{F}^{\sigma}_{\hat{\theta}; rs}(\cdot) & = \frac{1}{w_{rs}} \sum_{i = r}^sF_{\hat{\theta}_n(X_{n\sigma(i)})}(\cdot).
\end{align*}
We use $\pi$ to denote a permutation such that $\hat{\theta}_{n}(X_{n\pi(1)}) \leq \dots \leq \hat{\theta}_{n}(X_{n\pi(n)})$. The permutation $\pi$ is a function of $(X_{ni}, Y_{ni})_{i = 1}^n$ via $(X_{ni})_{i = 1}^n$ and $\hat{\theta}_n$. Let
\begin{equation} \label{eq:Mn}
	M_n^{\pi} = \max_{1 \leq r \leq s \leq n} w_{rs}^{1/2} \|\hat{\mathbb{F}}^{\pi}_{rs} - \bar{F}^{\pi}_{\theta; rs}\|_{\infty}.
\end{equation}

\begin{lem} \label{lem:Mn}
Under (A3) and (A4), there exists a constant $s = s(C_0, C_2) > 0$ such that
\[
	\lim_{n \rightarrow \infty} \mathbb{P}(M_n^{\pi} \geq sn^{1/4}\log(n)^{1/4}) = 0.
\]
\end{lem}
\begin{proof}
Define $m = m(n) = \max(1, \lfloor \lambda(I)/(2c_n) \rfloor)$ with $c_n = C_0 (\log(n)/n)^{1/2}$, where $C_0$ is from assumption (A4). Then, for $n$ large enough such that $c_n \leq \lambda(I)/4$,
\begin{equation} \label{eq:mn}
	2c_n \leq \frac{\lambda(I)}{m} \leq 4c_n.
\end{equation}
Slice the interval $I$ from (A3) into $m$ equally sized, disjoint intervals $J_1, \dots, J_m$ (ordered increasingly). Let $\mathcal{I}_k = \{i \in \{1, \dots, n\}: \theta(X_{ni}) \in J_k\}$, $n_k = \# \mathcal{I}_k$ for $k = 1, \dots, m$, and $N_n = \max_{k = 1, \dots, m} n_k$. Define also $\mathcal{I}_j = \emptyset$ for $j \not\in \{1, \dots, m\}$ and $\bigcup_{i = a}^b A_i = \emptyset$ for any sets $A_i$ and $a > b$.

Let $r, s \in \{1, \dots, n\}$, $r \leq s$, be indices that attain the maximum in \eqref{eq:Mn}, and define the index set $\mathcal{I}^* = \pi(\{r, \dots, s\})$, so that
\[
	M_n^{\pi} = \Big\|\frac{1}{(\#\mathcal{I}^*)^{1/2}} \sum_{i \in \mathcal{I}^*} \left(\one\{Y_{ni} \leq \cdot\} - F_{\theta(X_{ni})}(\cdot) \right) \Big\|_{\infty}.
\]
Note that the indices $r$ and $s$ are (complicated but measurable) functions of $(X_i, Y_i)$, $i = 1, \dots, n$, and thus random variables. Therefore, the set $\mathcal{I}^*$ is also a random set of indices.

If $i, j \in \mathcal{I}^*$ and $g(\hat{\theta}_n(X_{ni})) < g(\hat{\theta}_n(X_{nj}))$, with $g$ from (A4), then $k \in \mathcal{I}^*$ for all $k$ such that $g(\hat{\theta}_n(X_{ni})) < g(\hat{\theta}_n(X_{nk})) < g(\hat{\theta}_n(X_{nj}))$. This follows from $\hat{\theta}_{n}(X_{n\pi(1)}) \leq \dots \leq \hat{\theta}_{n}(X_{n\pi(n)})$, because if $i = \pi(i_0)$, $j = \pi(j_0)$ and $k = \pi(k_0)$, then $g(\hat{\theta}_n(X_{ni})) < g(\hat{\theta}_n(X_{nk})) < g(\hat{\theta}_n(X_{nj}))$ implies that $i_0 < k_0 < j_0$, and $k_0 \in \{i_0, \dots, j_0\} \subseteq \{r, \dots, s\}$ gives $k = \pi(k_0) \in \pi(\{r, \dots, s\}) = \mathcal{I}^*$.

Under the event $B_n$ defined at \eqref{eq:Bn}, $i \in \mathcal{I}_k$ and \eqref{eq:mn} imply that $g(\hat{\theta}(X_{ni})) \in J_t$ for some $t \in \{k-1, k, k + 1\}$. Therefore, for $l, k \in \{1, \dots, m\}$ with $l - k > 2$, it follows $g(\hat{\theta}(X_{ni})) < g(\hat{\theta}(X_{nj}))$ for all $i \in \mathcal{I}_{k}$ and $j \in \mathcal{I}_{l}$. So if $\mathcal{I}^*$ contains indices $i \in \mathcal{I}_k$ and $j \in \mathcal{I}_l$ with $l - k > 2$, then $\mathcal{I}^*$ must also contain all elements of the sets $\mathcal{I}_t$ for $k + 2 < t < l - 2$. Let $\kappa = \min\{j \in \{1, \dots, m\}: \, \mathcal{I}_j \cap \mathcal{I}^* \neq \emptyset\}$, $\ell = \max\{j \in \{1, \dots, m\}: \, \mathcal{I}_j \cap \mathcal{I}^* \neq \emptyset\}$. By the previous considerations, $\mathcal{I}^*$ may contain arbitrary elements of $\mathcal{I}_t$ with $t \in \{\kappa, \kappa + 1, \kappa + 2, \ell - 2, \ell - 1, \ell\}$, and it must contain all indices in $\mathcal{I}_j$ for $\kappa + 3 \leq j \leq \ell - 3$. In conclusion, under $B_n$, $\mathcal{I}^*$ is almost surely contained in the collection of index sets defined by
\[
	S_n = \bigcup_{1 \leq k \leq l \leq m} \left\{ \mathcal{J} \cup \left( \bigcup_{t = k + 3}^{l - 3} \mathcal{I}_t \right): \mathcal{J} \subseteq \left(\bigcup_{t = k}^{k + 2}\mathcal{I}_t\right) \cup \left(\bigcup_{t = l - 2}^{l}\mathcal{I}_t \right)  \,\right\}.
\]
Indeed, on the event $B_n$, we know that $\mathcal{I}^*$ must contain all elements of $\mathcal{I}_j$ for $\kappa + 3 \leq t \leq \ell - 3$. This explains the part $\bigcup_{t = k + 3}^{l - 3} \mathcal{I}_t$ in the definition of $S_n$. As for the $\mathcal{I}_k$ with subscript not in $\{\kappa + 3, \dots, \ell - 3\}$, $\mathcal{I}^*$ may contain any arbitrary selection from their elements. This arbitrary selection is $\mathcal{J} \subseteq \left(\bigcup_{t = k}^{k + 2}\mathcal{I}_t\right) \cup \left(\bigcup_{t = l - 2}^{l}\mathcal{I}_t \right)$. 
For $\kappa$ and $\ell$, all pairs $(k,l)$ with $k \leq l$ are possible, which gives the union over $1 \leq k \leq l \leq n$.

Because $\# \mathcal{I}_t \leq N_n$ for all $t$, one can derive from the definition of $S_n$ that
\begin{align*}
	\#S_n \leq m^2 \cdot 2^{6N_n} = m^2 \exp(6\log(2)N_n).
\end{align*}
We now compute an upper bound for $N_n$, which is a function of $\theta(X_{n1}), \dots, \theta(X_{nn})$ only. Denote by $P$ and $G$ the distribution and the CDF of $\theta(X_{n1})$, and by $\hat{P}$ and $\hat{G}$ the empirical distribution and the empirical CDF of $\theta(X_{n1}), \dots, \theta(X_{nn})$. For any $c \geq 0$,
\begin{align*}
	\mathbb{P}(N_n \geq c) & \leq \sum_{k = 1}^m \mathbb{P}(n_k \geq c) \\
	& \leq \sum_{k = 1}^m \mathbb{P}\Big(\hat{P}(J_k) - P(J_k) \geq \frac{c}{n} - P(J_k)\Big) \\
	& \leq \sum_{k = 1}^m \mathbb{P}\Big(2\|G - \hat{G}\|_{\infty} \geq \frac{c}{n} - P(J_k)\Big).
\end{align*}
For $n$ sufficiently large, $P(J_k) \leq 4C_2C_0c_n = 4C_2C_0(\log(n)/n)^{1/2}$ by \eqref{eq:mn} and by (A3). Replacing $c$ by $d_n = R\log(n)^{1/2}n^{1/2}$ with $R = \max(2, 8C_2C_0)$ and applying Lemma \ref{lem:lem1} and \eqref{eq:mn} yields
\begin{align*}
	\mathbb{P}(N_n \geq d_n) & \leq \sum_{k = 1}^m \mathbb{P}\Big(2\|G - \hat{G}\|_{\infty} \geq \frac{d_n}{n} - 4C_2C_0(\log(n)/n)^{1/2}\Big) \\
	& \leq \sum_{k = 1}^m \mathbb{P}\Big(2\|G - \hat{G}\|_{\infty} \geq \frac{d_n}{2n}\Big) \\
	& \leq m M\exp \left(-2n\Big(\frac{d_n}{4n}\Big)^2 \right) \\
	& \leq \frac{\lambda(I)M}{2(\log(n)/n)^{1/2}}\exp\left(-\log(n)/2\right) \\
	& \leq \frac{\lambda(I)M}{2\log(n)^{1/2}}\exp\left(-\log(n)/2 + \log(n)/2\right) \rightarrow 0, \ n \rightarrow \infty.
\end{align*}

So with asymptotic probability one,
\begin{align*}
\#S_n \leq m^2 \exp\left(6 \log(2) R \log(n)^{1/2}n^{1/2}\right) & \leq \frac{\lambda(I)^2}{4c_n^2} \exp\left(6R \log(2) \log(n)^{1/2}n^{1/2}\right) \\
& \leq r_0 \exp\left(r_1 \log(n)^{1/2}n^{1/2} \right),
\end{align*}
with $r_0 = \lambda(I)^2/(4C_0)$ and $r_1 = 6R\log(2) + 1$.
Define $D_n = \{\#S_n \leq r_0\exp(r_1 \log(n)^{1/2}n^{1/2})\}$, let $\mathfrak{S}_n$ be the power set of $\{1, \dots, n\}$, and, for $\mathcal{J} \in \mathfrak{S}_n$,
\[
	M_n^{\mathcal{J}} = \Big\|\frac{1}{(\#\mathcal{J})^{1/2}} \sum_{i \in \mathcal{J}} \left( \one\{Y_{ni} \leq \cdot\} - F_{\theta(X_{ni})}(\cdot)\right)\Big\|_{\infty}.
\]
Then, for $z_n = s\log(n)^{1/4}n^{1/4}$ with an arbitrary $s > 0$,
\begin{align*}
	& \mathbb{P}(M_n^{\pi} \geq z_n) = \mathbb{E}\left(\one\left\{M_n^{\mathcal{I}^*} \geq z_n\right\} \right) \\
	& \ = \mathbb{E}\left(\sum_{\mathcal{J} \in \mathfrak{S}_n} \one\{\mathcal{I}^* = \mathcal{J}\} \one\left\{M_n^{\mathcal{J}} \geq z_n\right\}\right) \\
	& \ \leq \mathbb{P}(B_n^c) + \mathbb{E}\left( \one B_n \sum_{\mathcal{J} \in \mathfrak{S}_n}\one\{\mathcal{I}^* = \mathcal{J}\} \one\left\{M_n^{\mathcal{J}} \geq z_n\right\}\right) \\
	& \ = \mathbb{P}(B_n^c) + \mathbb{E}\left( \one B_n \sum_{\mathcal{J} \in S_n}\one\{\mathcal{I}^* = \mathcal{J}\}  \one\left\{M_n^{\mathcal{J}} \geq z_n\right\}\right) \\
	& \ \leq \mathbb{P}(B_n^c) + \mathbb{E}\left(\sum_{\mathcal{J} \in S_n}\one\{\mathcal{I}^* = \mathcal{J}\} \one\left\{M_n^{\mathcal{J}} \geq z_n\right\}\right) \\
	& \ \leq \mathbb{P}(B_n^c) + \mathbb{P}(D_n^c) + \mathbb{E}\left( \one D_n \mathbb{E}\left[\sum_{\mathcal{J} \in S_n} \one\{\mathcal{I}^* = \mathcal{J}\} \one\left\{M_n^{\mathcal{J}} \geq z_n\right\} \middle| X_{n1}, \dots, X_{nn} \right]\right).
\end{align*}
In the last inequality we use the fact that $\one D_n$ is a function of $X_{n1}, \dots, X_{nn}$ and
\[
	\mathbb{E}\left[\sum_{\mathcal{J} \in S_n} \one\{\mathcal{I}^* = \mathcal{J}\} \one\left\{M_n^{\mathcal{J}} \geq z_n\right\} \middle| X_{n1}, \dots, X_{nn} \right] \leq 1 \text{ a.s.},
\]
since $\mathcal{I}^* = \mathcal{J}$ may only hold for exactly one index set $\mathcal{J}$. Finally,
\begin{align*}
& \mathbb{E}\left( \one D_n \mathbb{E}\left[\sum_{\mathcal{J} \in S_n} \one\{\mathcal{I}^* = \mathcal{J}\} \one\left\{M_n^{\mathcal{J}} \geq z_n\right\} \middle| X_{n1}, \dots, X_{nn} \right]\right) \\
& \leq \mathbb{E}\left( \one D_n \sum_{\mathcal{J} \in S_n} \mathbb{E}\left[\one\left\{M_n^{\mathcal{J}} \geq z_n\right\} \mid X_{n1}, \dots, X_{nn} \right]\right). \\
& = \mathbb{E}\left( \one D_n \sum_{\mathcal{J} \in S_n} \mathbb{P}\left[M_n^{\mathcal{J}} \geq z_n \mid X_{n1}, \dots, X_{nn} \right]\right). \\
& \leq \mathbb{E}\left( \one D_n (\# S_n) M\exp(-2z_n^2)\right) \\
& \leq r_0M\exp\left(-(2s^2 - r_1)\log(n)^{1/2}n^{1/2} \right) \rightarrow 0, \ n \rightarrow \infty,
\end{align*}
for $s > \sqrt{r_1/2}$, using Lemma \ref{lem:lem1} in the second-last inequality.
\end{proof}

Lemma \ref{lem:dense} shows that for suitable constants $D$ and sequences $(\delta_n)_{n \in \mathbb{N}}$ with limit zero, all subintervals of $I$ with length at least $\delta_n$ contain at least $Dn\delta_n$ elements of $\{g(\hat{\theta}_n(X_{nj})): \, j = 1, \dots, n\}$. That is, the pseudo-covariates $g(\hat{\theta}_n(X_{nj}))$ are asymptotically dense in $I$.

\begin{lem} \label{lem:dense}
Under (A3) and (A4), with $\hat{w}(B) = \#\{j \in \{1, \ldots, n\}: g(\hat{\theta}_n(X_{nj})) \in B\}$, for any sequence $(\delta_n)_{n \in \mathbb{N}}$ such that $\delta_n \geq 4C_0 (\log(n)/n)^{1/2}$, the event
\begin{equation} \label{eq:dense2}
	\left\{\inf \left\{\frac{\hat{w}(I_n)}{n\lambda(I_n)}: \text{intervals } I_n \subset I \text{ with } \lambda(I_n) \geq \delta_n \right\} \geq D\right\}
\end{equation}
has asymptotic probability one for any $D < C_1/2$.
\end{lem}

\begin{proof}[Proof of Lemma \ref{lem:dense}]
Similarly to the definition of $\hat{w}$, let $w(B) = \#\{j \in \{1, \ldots, n\}: \theta(X_{nj}) \in B\}$ for $B \subseteq I$. Define $c_n = C_0(\log(n)/n)^{1/2}$ with $C_0$ from (A4). Then on the event $B_n$ defined at \eqref{eq:Bn}, for any interval $J \subseteq I$ with $\lambda(J) \geq 2c_n$,
\begin{align*}
	\hat{w}(J)-w(J) & \geq -\#\{j \in \{1, \dots, n\}: \, \hat{\theta}_n(X_{nj}) \not\in J, \, \theta(X_{nj}) \in J\} \\
	& \geq -w(\{z \in J: \, z + c_n \not\in J \text{ or } z - c_n \not\in J \}).
\end{align*}
This gives $\hat{w}(J) \geq \ w(J \setminus \{z \in J: \, z + c_n \not\in J \text{ or } z - c_n \not\in J \})$. The assumption $\delta_n \geq 4c_n$ implies that $\delta_n - 2c_n \geq \delta_n / 2$. For any interval $I_n \subseteq I$ of length at least $\delta_n$, the set $\tilde{I}_n = I_n \setminus \{z \in I_n: \, z + c_n \not\in I_n \text{ or } z - c_n \not\in I_n \}$ is an interval of length
\[
	\lambda(\tilde{I}_n) = \lambda(I_n) - 2c_n \geq \lambda(I_n) - \delta_n/2 \geq \lambda(I_n) - \lambda(I_n)/2 = \lambda(I_n)/2.
\]
This and $\hat{w}(I_n)\geq w(\tilde{I}_n)$ yield
\begin{align*}
	\hat{m}_n & := \inf \left\{\frac{\hat{w}(I_n)}{n\lambda(I_n)}: \text{intervals } I_n \subset I \text{ with } \lambda(I_n) \geq \delta_n \right\} \\
	& \geq \inf \left\{\frac{w(\tilde{I}_n)}{n\lambda(\tilde{I}_n)}: \text{intervals } \tilde{I}_n \subset I \text{ with } \lambda(\tilde{I}_n) \geq \delta_n/2 \right\}/2 =: m_n.
\end{align*}
Define $A_n = \{\hat{m}_n \geq D\}$ and $\tilde{A}_n = \{m_n \geq D\}$ for $D < C_1/2$. Then $\tilde{A}_n \subseteq A_n$ and
\[
	\mathbb{P}(A_n) \geq \mathbb{P}(A_n \cap B_n) \geq \mathbb{P}(\tilde{A}_n \cap B_n) = \mathbb{P}(\tilde{A}_n) + \mathbb{P}(B_n)- \mathbb{P}(\tilde{A}_n \cup B_n) \rightarrow 1, \, n \rightarrow \infty,
\]
since $\lim_{n\rightarrow\infty}\mathbb{P}(B_n) = 1$ by (A4) and $\lim_{n\rightarrow\infty}\mathbb{P}(\tilde{A}_n) = 1$ by (A3) and by Equation 4.6 of \citet[Section 4.3]{Moesching2020}.
\end{proof}

\begin{proof}[Proof of Theorem \ref{thm:consistency}]
Proposition \ref{prop:invariance} implies that for all $u \in \mathbb{R}$,
\[
	\hat{F}_u(y; (\hat{\theta}_n(X_{nj}))_{j = 1}^n, (Y_{nj})_{j = 1}^n) = \hat{F}_{g(u)}(y; (g(\hat{\theta}_n(X_{nj})))_{j = 1}^n, (Y_{nj})_{j = 1}^n).
\]
To lighten the notation, we can therefore drop $g$ from (A4) and simply write $\hat{\theta}_n(\cdot)$ instead of $g(\hat{\theta}_n(\cdot))$. Assume that $\hat{\theta}_n(X_{n\pi(1)}) \leq \hat{\theta}_n(X_{n\pi(2)}) \leq \ldots \leq \hat{\theta}_n(X_{n\pi(n)})$ and define $\delta_n = (\log n / n)^{1/6}/2$. Lemma \ref{lem:dense} and (A4) imply that for all $x \in \mathcal{X}_n = \{x \in \mathcal{X}: \, [\theta(x) \pm 2\delta_n ] \subseteq I\}$, the indices
\begin{align*}
	r(x) & = \min\{j \in \{1, \ldots, n\}: \hat{\theta}_n(X_{n\pi(j)}) \geq \hat{\theta}_n(x) - \delta_n\} \\
	j(x) & = \max\{j \in \{1, \ldots, n\}: \hat{\theta}_n(X_{n\pi(j)}) \leq \hat{\theta}_n(x) \}
\end{align*}
are well defined with asymptotic probability one, because $[\hat{\theta}_n(x) - \delta_n, \hat{\theta}_n(x)]$ is of length $\delta_n$ and contained in $I$ since $\theta(x) + (\log n / n)^{1/6} \geq \hat{\theta}_n(x) \geq \hat{\theta}_n(x) - \delta_n \geq \theta(x) - \delta_n - C_0n^{-1/2} > \theta(x) - (\log n / n)^{1/6}$ for $n$ sufficiently large, on the event $B_n$ defined at \eqref{eq:Bn}. They satisfy $r(x) \leq j(x)$ and $\hat{\theta}_n(x) - \delta_n \leq \hat{\theta}_n(X_{nr(x)}) \leq \hat{\theta}_n(X_{nj(x)}) \leq \hat{\theta}_n(x)$ and, with asymptotic probability one due to Lemma \ref{lem:dense}, $w_{r(x)j(x)} = \#\{j \in \{1, \dots, n\}: \hat{\theta}_n(x) - \delta_n \leq \hat{\theta}_n(X_{n\pi(j)}) \leq \hat{\theta}_n(x)\} \geq D n \delta_n$ for $0 < D < C_1/2$.
Therefore, almost surely with respect to the joint law of $(X_{ni}, Y_{ni})$, $i = 1, \dots, n$, for any $y \in \mathbb{R}$,
\begin{align*}
	\hat{F}_{n; \hat{\theta}_n(x)} (y) - F_{\theta(x)}(y) & \leq 
	\hat{F}_{n; \hat{\theta}_n(X_{nj(x)})}(y) -  F_{\theta(x)}(y) \\
	& = \min_{r \leq j(x)} \max_{s \geq j(x)} \hat{\mathbb{F}}^{\pi}_{rs}(y) - F_{\theta(x)}(y) \\
	& \leq \max_{s \geq j(x)} \hat{\mathbb{F}}^{\pi}_{r(x)s}(y)  - F_{\theta(x)}(y) \\
	& \leq w_{r(x)j(x)}^{-1/2} M_n^{\pi} + \max_{s \geq j(x)} \bar{F}^{\pi}_{\theta; r(x)s}(y) - F_{\theta(x)}(y) \\
	& \leq (Dn\delta_n)^{-1/2}M_n^{\pi} \\
	& \qquad + \max_{s \geq j(x)} \big(\bar{F}^{\pi}_{\theta; r(x)s}(y) - \bar{F}^{\pi}_{\hat{\theta}; r(x)s}(y) + \bar{F}^{\pi}_{\hat{\theta}; r(x)s}(y) \big) - F_{\theta(x)}(y) \\
	& \leq (Dn\delta_n)^{-1/2}M_n^{\pi} + L\sup_{x \in \mathcal{X}}|\hat{\theta}_n(x) - \theta(x)| + \max_{s \geq j(x)} \bar{F}^{\pi}_{\hat{\theta}; r(x)s}(y) - F_{\theta(x)}(y) \\
	& \leq (Dn\delta_n)^{-1/2}M_n^{\pi} + L\sup_{x \in \mathcal{X}}|\hat{\theta}_n(x) - \theta(x)| + F_{\hat{\theta}_n(X_{nr(x)})}(y) - F_{\theta(x)}(y) \\
	& \leq (Dn\delta_n)^{-1/2}M_n^{\pi} + L\sup_{x \in \mathcal{X}}|\hat{\theta}_n(x) - \theta(x)| + L|\hat{\theta}_n(X_{nr(x)}) - \theta(x)| \\
	& \leq (Dn\delta_n)^{-1/2}M_n^{\pi} + L\sup_{x \in \mathcal{X}}|\hat{\theta}_n(x) - \theta(x)| + L\delta_n.
\end{align*}
The equality in the second line is the classical min-max formula for monotone regression, see e.g.~Equation (2.2) in \cite{Moesching2020}, and the first and the third last inequality use antitonicity of $u \mapsto F_u(y)$. By assumption (A4) and with the constant $s > 0$ from Lemma \ref{lem:Mn}, the event
\[
	\{M_n^{\pi} \leq s(n\log(n))^{1/4}\} \cap  \left\{\sup_{x \in \mathcal{X}}|\hat{\theta}_n(x) - \theta(x)| < \delta_n\right\}
\]
has asymptotic probability one. On this event, the previous considerations imply
\[
	\sup_{x \in \mathcal{X}_n, y \in \mathbb{R}} (\hat{F}_{n; \hat{\theta}_n(x)}(y) - F_{\theta(x)})(y) \leq s(Dn\delta_n)^{-1/2}(n \log(n))^{1/4} + 2L\delta_n
	\leq C \left(\frac{\log(n)}{n} \right)^{1/6},
\]
with $C = [s(2D^{-1})^{1/2} + L]$. To finish the proof, we show that $F_{\theta(x)}(y) - \hat{F}_{n; \hat{\theta}_n(x)}(y)$ can be bounded in the same way.

Similar to before, define the indices $r'(x) = \min\{j \in \{1, \ldots, n\}: \hat{\theta}_n(X_{nj}) \geq \hat{\theta}_n(x)\}$, $j'(x) = \max\{j \in \{1, \ldots, n\}: \hat{\theta}_n(X_{nj}) \leq \hat{\theta}_n(x) + \delta_n \}$.
Then with asymptotic probability one, also $r'(x) \leq j'(x)$ and $\hat{\theta}_n(x) \leq \hat{\theta}_n(X_{nr'(x)}) \leq \hat{\theta}_n(X_{nj'(x)}) \leq \hat{\theta}_n(x) + \delta_n$, $w_{r'(x)j'(x)} \geq D n \delta_n$.
Thus,
\begin{align*}
	\hat{F}_{n; \hat{\theta}_n(x)}(y) - F_{\theta(x)}(y) & \geq 
	\hat{F}_{n; \hat{\theta}_n(X_{nr'(x)})}(y) - F_{\theta(x)} \\
	& = \min_{r \leq r'(x)} \max_{s \geq r'(x)} \hat{\mathbb{F}}^{\pi}_{rs}(y) - F_{\theta(x)}(y) \\
	& \geq \min_{r \leq r'(x)} \hat{\mathbb{F}}^{\pi}_{rj'(x)}(y) - F_{\theta(x)}(y) \\
	& \geq -w_{r'(x)j'(x)}^{-1/2} M_n^{\pi} + \min_{r \leq r'(x)} \bar{F}^{\pi}_{\theta; rj'(x)}(y) - F_{\theta(x)}(y) \\	
	& \geq -(Dn\delta_n)^{-1/2}M_n^{\pi}  \\
	& \qquad + \min_{r \leq r'(x)} \big( \bar{F}^{\pi}_{\theta; rj'(x)}(y) - \bar{F}^{\pi}_{\hat{\theta}; rj'(x)}(y) + \bar{F}^{\pi}_{\hat{\theta}; rj'(x)}(y) \big) - F_{\theta(x)}(y) \\	
	& \geq -(Dn\delta_n)^{-1/2}M_n^{\pi} - L\sup_{x \in \mathcal{X}}|\hat{\theta}_n(x) - \theta(x)| + F_{\hat{\theta}_n(X_{nj'(x)})}(y) - F_{\theta(x)}(y) \\
	& \geq -(Dn\delta_n)^{-1/2}M_n^{\pi} - L\sup_{x \in \mathcal{X}}|\hat{\theta}_n(x) - \theta(x)| - L|\hat{\theta}_n(X_{nj'(x)}) - \theta(x)| \\
	& \geq -(Dn\delta_n)^{-1/2}M_n^{\pi} - L\sup_{x \in \mathcal{X}}|\hat{\theta}_n(x) - \theta(x)| - L\delta_n.\qedhere
\end{align*}
\end{proof}

\begin{proof}[Proof of Theorem \ref{thm:consistency} with sample splitting]
Assume that the index estimator $\hat{\theta}_n$ is computed with data $(X_{ni}, Y_{ni})_{i = 1}^{\lfloor n\xi\rfloor}$ and the distribution functions with $(\hat{\theta}_n(X_{ni}), Y_{ni})_{i = \lfloor n\xi \rfloor + 1}^n$. The statement of Lemma \ref{lem:dense} also holds when $C_0(\log(n)/n)^{1/2}$ is replaced by $(\log(n)/n)^{1/3}$. By conditioning on $(X_{ni}, Y_{ni})_{i = 1}^{\lfloor n\xi\rfloor}$ and on $X_{ni}$, $i = \lfloor n\xi \rfloor + 1, \dots, n$, Corollary 4.7 of \cite{Moesching2020} implies that $M_n^{\pi}$ (computed with the data $(\hat{\theta}_n(X_{ni}), Y_{ni})_{i = \lfloor n\xi \rfloor + 1}^n$) satisfies $\mathbb{P}(M_n^{\pi} \geq (R\log(n(1-\xi)))^{1/2}) \rightarrow 0$, $n \rightarrow \infty$, for any $R > 1$. This requires the fact that the permutation $\pi$ is constant when conditioned on $(X_{ni}, Y_{ni})_{i = 1}^{\lfloor n\xi \rfloor}$. One may now follow exactly the same steps as in the proof for the theorem without sample splitting, but with sample size $\lfloor n(1-\xi) \rfloor$ instead of $n$, $\delta_n = (n(1-\xi)/\log(n(1-\xi)))^{1/3}/2$ instead of $(n/\log(n))^{1/6}$ and $\{M_n^{\pi} \leq (R\log(n(1-\xi)))^{1/2}\}$ instead of $\{M_n^{\pi} \leq s(n\log(n))^{1/4}\}$, obtaining an upper bound of $C'(\log(n)/n)^{1/3}$ for the error, where $C' > 0$ also depends on $\xi$.
\end{proof}

\bibliographystyle{apalike} 
\bibliography{biblio.DIM}

\begin{thebibliography}{}

\bibitem[Athey et~al., 2019]{Athey2019}
Athey, S., Tibshirani, J., and Wager, S. (2019).
\newblock Generalized random forests.
\newblock {\em Annals of Statistics}, 37:1148--1178.

\bibitem[Balabdaoui et~al., 2019a]{Balabdaoui2019}
Balabdaoui, F., Durot, C., and Jankowski, H. (2019a).
\newblock Least squares estimation in the monotone single index model.
\newblock {\em Bernoulli}, 25:3276--3310.

\bibitem[{Balabdaoui} and {Groeneboom}, 2020]{Balabdaoui2020}
{Balabdaoui}, F. and {Groeneboom}, P. (2020).
\newblock {Profile least squares estimators in the monotone single index
  model}.
\newblock {\em arXiv e-prints}, page arXiv:2001.05454.

\bibitem[Balabdaoui et~al., 2019b]{Balabdaoui2019a}
Balabdaoui, F., Groeneboom, P., and Hendrickx, K. (2019b).
\newblock Score estimation in the monotone single-index model.
\newblock {\em Scandinavian Journal of Statistics}, 46:517--544.

\bibitem[Carroll et~al., 1997]{Carroll1997}
Carroll, R.~J., Fan, J., Gijbels, I., and Wand, M.~P. (1997).
\newblock Generalized partially linear single-index models.
\newblock {\em Journal of the American Statistical Association}, 92:477--489.

\bibitem[Chernozhukov et~al., 2010]{Chernozhukov2010}
Chernozhukov, V., Fern{\'a}ndez-Val, I., and Galichon, A. (2010).
\newblock Quantile and probability curves without crossing.
\newblock {\em Econometrica}, 78:1093--1125.

\bibitem[Chernozhukov et~al., 2013]{Chernozhukov2013}
Chernozhukov, V., Fern{\'a}ndez-Val, I., and Melly, B. (2013).
\newblock Inference on counterfactual distributions.
\newblock {\em Econometrica}, 81:2205--2268.

\bibitem[Chernozhukov et~al., 2020]{Chernozhukov2020}
Chernozhukov, V., Fern{\'a}ndez-Val, I., Melly, B., and W{\"u}thrich, K.
  (2020).
\newblock Generic inference on quantile and quantile effect functions for
  discrete outcomes.
\newblock {\em Journal of the American Statistical Association}, 115:123--137.

\bibitem[Cox, 1972]{Cox1972}
Cox, D.~R. (1972).
\newblock Regression models and life-tables.
\newblock {\em Journal of the Royal Statistical Society: Series B},
  34:187--202.

\bibitem[Dette and Volgushev, 2008]{DetteVolgushev2008}
Dette, H. and Volgushev, S. (2008).
\newblock Non-crossing non-parametric estimates of quantile curves.
\newblock {\em Journal of the Royal Statistical Society: Series B},
  70:609--627.

\bibitem[Diebold et~al., 1998]{DieboldGuntherETAL1998}
Diebold, F.~X., Gunther, T.~A., and Tay, A.~S. (1998).
\newblock Evaluating density forecasts with applications to financial risk
  management.
\newblock {\em International Economic Review}, 39:863--883.

\bibitem[Duarte et~al., 2017]{Duarte2017}
Duarte, E., de~Sousa, B., Cadarso-Su{\'a}rez, C., Klein, N., Kneib, T., and
  Rodrigues, V. (2017).
\newblock Studying the relationship between a woman's reproductive lifespan and
  age at menarche using a {Bayesian} multivariate structured additive
  distributional regression model.
\newblock {\em Biometrical Journal}, 59:1232--1246.

\bibitem[Dunson et~al., 2007]{Dunson2007}
Dunson, D.~B., Pillai, N., and Park, J.-H. (2007).
\newblock Bayesian density regression.
\newblock {\em Journal of the Royal Statistical Society: Series B},
  69:163--183.

\bibitem[Foresi and Peracchi, 1995]{Foresi1995}
Foresi, S. and Peracchi, F. (1995).
\newblock The conditional distribution of excess returns: An empirical
  analysis.
\newblock {\em Journal of the American Statistical Association}, 90:451--466.

\bibitem[Gneiting et~al., 2007]{Gneiting2007}
Gneiting, T., Balabdaoui, F., and Raftery, A.~E. (2007).
\newblock Probabilistic forecasts, calibration and sharpness.
\newblock {\em Journal of the Royal Statistical Society: Series B},
  69:243--268.

\bibitem[Gneiting and Katzfuss, 2014]{Gneiting2014}
Gneiting, T. and Katzfuss, M. (2014).
\newblock Probabilistic forecasting.
\newblock {\em Annual Review of Statistics and Its Application}, 1:125--151.

\bibitem[Gneiting and Raftery, 2007]{Gneiting2007a}
Gneiting, T. and Raftery, A.~E. (2007).
\newblock Strictly proper scoring rules, prediction, and estimation.
\newblock {\em Journal of the American Statistical Association}, 102:359--378.

\bibitem[{Gneiting} and {Walz}, 2019]{Gneiting2019}
{Gneiting}, T. and {Walz}, E.-M. (2019).
\newblock {Receiver operating characteristic (ROC) movies, universal ROC (UROC)
  curves, and coefficient of predictive ability (CPA)}.
\newblock {\em arXiv e-prints}.

\bibitem[Hall et~al., 1999]{Hall1999}
Hall, P., Wolff, R. C.~L., and Yao, Q. (1999).
\newblock Methods for estimating a conditional distribution function.
\newblock {\em Journal of the American Statistical Association}, 94:154--163.

\bibitem[Hall and Yao, 2005]{Hall2005}
Hall, P. and Yao, Q. (2005).
\newblock Approximating conditional distribution functions using dimension
  reduction.
\newblock {\em Annals of Statistics}, 33:1404--1421.

\bibitem[H\"ardle et~al., 1993]{Hardle1993}
H\"ardle, W., Hall, P., and Ichimura, H. (1993).
\newblock Optimal smoothing in single-index models.
\newblock {\em The Annals of Statistics}, 21:157--178.

\bibitem[Hastie and Tibshirani, 1990]{Hastie1990}
Hastie, T.~J. and Tibshirani, R.~J. (1990).
\newblock {\em Generalized additive models}, volume~43 of {\em Monographs on
  Statistics and Applied Probability}.
\newblock Chapman and Hall, Ltd., London.

\bibitem[Hayfield and Racine, 2008]{Hayfield2008}
Hayfield, T. and Racine, J.~S. (2008).
\newblock Nonparametric econometrics: The np package.
\newblock {\em Journal of Statistical Software}, 27:1--32.

\bibitem[Henzi et~al., 2020]{Henzi2020}
Henzi, A., M\"{o}sching, A., and D\"{u}mbgen, L. (2020).
\newblock Accelerating the pool-adjacent-violators algorithm for isotonic
  distributional regression.
\newblock Preprint, \url{arxiv.org/abs/2006.05527}.

\bibitem[{Henzi} et~al., 2019]{Henzi2019}
{Henzi}, A., {Ziegel}, J.~F., and {Gneiting}, T. (2019).
\newblock Isotonic distributional regression.
\newblock {\em arXiv e-prints}, page arXiv:1909.03725.

\bibitem[Hot\-horn et~al., 2014]{Hothorn2014}
Hot\-horn, T., Kneib, T., and B\"uhl\-mann, P. (2014).
\newblock Conditional transformation models.
\newblock {\em Journal of the Royal Statistical Society: Series B}, 76:3--27.

\bibitem[{Jordan} et~al., 2019]{Jordan2019}
{Jordan}, A.~I., {M{\"u}hlemann}, A., and {Ziegel}, J.~F. (2019).
\newblock Optimal solutions to the isotonic regression problem.
\newblock {\em arXiv e-prints}, page arXiv:1904.04761.

\bibitem[Klein et~al., 2015]{Klein2015}
Klein, N., Kneib, T., Lang, S., and Sohn, A. (2015).
\newblock Bayesian structured additive distributional forecasting with an
  application to regional income inequality in {G}ermany.
\newblock {\em Annals of Applied Statistics}, 9:1024--1052.

\bibitem[Koenker, 2005]{Koenker2005}
Koenker, R. (2005).
\newblock {\em Quantile Regression}.
\newblock Cambridge University Press.

\bibitem[Koenker, 2020]{Koenker2020}
Koenker, R. (2020).
\newblock {\em quantreg: Quantile Regression}.
\newblock R package version 5.55.

\bibitem[Kramer, 2017]{Kramer2017}
Kramer, A.~A. (2017).
\newblock Are {ICU} length of stay predictions worthwhile?
\newblock {\em Critical Care Medicine}, 45:379--380.

\bibitem[{Kuchibhotla} et~al., 2017]{Kuchibhotla2017}
{Kuchibhotla}, A.~K., {Patra}, R.~K., and {Sen}, B. (2017).
\newblock Least squares estimation in a single index model with convex
  {L}ipschitz link.
\newblock {\em arXiv e-prints}, page arXiv:1708.00145.

\bibitem[{Lanteri} et~al., 2020]{Lanteri2020}
{Lanteri}, A., {Maggioni}, M., and {Vigogna}, S. (2020).
\newblock Conditional regression for single-index models.
\newblock {\em arXiv e-prints}, page arXiv:2002.10008.

\bibitem[Le~Gall et~al., 1993]{le1993}
Le~Gall, J.-R., Lemeshow, S., and Saulnier, F. (1993).
\newblock A new simplified acute physiology score ({SAPS II}) based on a
  {European/North American} multicenter study.
\newblock {\em JAMA}, 270:2957--2963.

\bibitem[Li and Racine, 2008]{Li2008}
Li, Q. and Racine, J.~S. (2008).
\newblock Nonparametric estimation of conditional {CDF} and quantile functions
  with mixed categorical and continuous data.
\newblock {\em Journal of Business \& Economic Statistics}, 26:423--434.

\bibitem[Machado and Mata, 2000]{Machado2000}
Machado, J. A.~F. and Mata, J. (2000).
\newblock Box--{Cox} quantile regression and the distribution of firm sizes.
\newblock {\em Journal of Applied Econometrics}, 15:253--274.

\bibitem[Matheson and Winkler, 1976]{Matheson1976}
Matheson, J.~E. and Winkler, R.~L. (1976).
\newblock Scoring rules for continuous probability distributions.
\newblock {\em Management Science}, 22:1087--1096.

\bibitem[McCullagh and Nelder, 1989]{McCullagh1989}
McCullagh, P. and Nelder, J.~A. (1989).
\newblock {\em Generalized linear models}.
\newblock Monographs on Statistics and Applied Probability. Chapman \& Hall,
  London, 2nd edition.

\bibitem[Meins\-hausen, 2006]{Meinshausen2006}
Meins\-hausen, N. (2006).
\newblock Quantile regression forests.
\newblock {\em Journal of Machine Learning Research}, 7:983--999.

\bibitem[Miranda et~al., 1997]{Miranda1997}
Miranda, D.~R., Moreno, R., and Iapichino, G. (1997).
\newblock Nine equivalents of nursing manpower use score ({NEMS}).
\newblock {\em Intensive Care Medicine}, 23:760--765.

\bibitem[Moran and Solomon, 2012]{Moran2012}
Moran, J.~L. and Solomon, P.~J. (2012).
\newblock A review of statistical estimators for risk-adjusted length of stay:
  analysis of the {Australian} and new {Zealand} intensive care adult patient
  data-base, 2008--2009.
\newblock {\em BMC Medical Research Methodology}, 12:68.

\bibitem[M\"osching and D\"umbgen, 2020]{Moesching2020}
M\"osching, A. and D\"umbgen, L. (2020).
\newblock Monotone least squares and isotonic quantiles.
\newblock {\em Electronic Journal of Statistics}, 14:24--49.

\bibitem[Niskanen et~al., 2009]{Niskanen2009}
Niskanen, M., Reinikainen, M., and Pettil{\"a}, V. (2009).
\newblock Case-mix-adjusted length of stay and mortality in 23 {Finnish}
  {ICU}s.
\newblock {\em Intensive Care Medicine}, 35:1060--1067.

\bibitem[Peracchi, 2002]{Peracchi2002}
Peracchi, F. (2002).
\newblock On estimating conditional quantiles and distribution functions.
\newblock {\em Computational Statistics \& Data Analysis}, 38:433--447.

\bibitem[{R Core Team}, 2020]{R}
{R Core Team} (2020).
\newblock {\em R: A Language and Environment for Statistical Computing}.
\newblock R Foundation for Statistical Computing, Vienna, Austria.

\bibitem[Rasp and Lerch, 2018]{Rasp2018}
Rasp, S. and Lerch, S. (2018).
\newblock Neural networks for postprocessing ensemble weather forecasts.
\newblock {\em Monthly Weather Review}, 146:3885--3900.

\bibitem[Rigby and Stasinopoulos, 2005]{Rigby2005}
Rigby, R.~A. and Stasinopoulos, D.~M. (2005).
\newblock Generalized additive models for location, scale and shape (with
  discussion).
\newblock {\em Journal of the Royal Statistical Society: Series C},
  54:507--554.

\bibitem[Schlosser et~al., 2019]{Schlosser2019}
Schlosser, L., Hothorn, T., Stauffer, R., and Zeileis, A. (2019).
\newblock Distributional regression forests for probabilistic precipitation
  forecasting in complex terrain.
\newblock {\em Annals of Applied Statistics}, 13:1564--1589.

\bibitem[Shaked and Shanthikumar, 2007]{Shaked2007}
Shaked, M. and Shanthikumar, J.~G. (2007).
\newblock {\em Stochastic Orders}.
\newblock Springer, New York.

\bibitem[Silbersdorff et~al., 2018]{Silbersdorff2018}
Silbersdorff, A., Lynch, J., Klasen, S., and Kneib, T. (2018).
\newblock Reconsidering the income-health relationship using distributional
  regression.
\newblock {\em Health Economics}, 27:1074--1088.

\bibitem[Thomas et~al., 2018]{Thomas2018}
Thomas, J., Mayr, A., Bischl, B., Schmid, M., Smith, A., and Hofner, B. (2018).
\newblock Gradient boosting for distributional regression: faster tuning and
  improved variable selection via noncyclical updates.
\newblock {\em Statistics and Computing}, 28:673--687.

\bibitem[Umlauf et~al., 2018]{Umlauf2018}
Umlauf, N., Klein, N., and Zeileis, A. (2018).
\newblock Bamlss: Bayesian additive models for location, scale, and shape (and
  beyond).
\newblock {\em Journal of Computational and Graphical Statistics}, 27:612--627.

\bibitem[Vannitsem et~al., 2018]{Vannitsem2018}
Vannitsem, S., Wilks, D.~S., and Messner, J., editors (2018).
\newblock {\em Statistical Postprocessing of Ensemble Forecasts}.
\newblock Elsevier.

\bibitem[Verburg et~al., 2014]{Verburg2014}
Verburg, I.~W., de~Keizer, N.~F., de~Jonge, E., and Peek, N. (2014).
\newblock Comparison of regression methods for modeling intensive care length
  of stay.
\newblock {\em PloS one}, 9(10):e109684.

\bibitem[Wilks, 2011]{Wilks2011}
Wilks, D.~S. (2011).
\newblock {\em Statistical {M}ethods in the {A}tmospheric {S}ciences}.
\newblock Elsevier, 3rd edition.

\bibitem[Wood, 2017]{Wood2017}
Wood, S.~N. (2017).
\newblock {\em Generalized Additive Models: An Introduction with R}.
\newblock Chapman and Hall/CRC, 2nd edition.

\bibitem[Zhang et~al., 2017]{Zhang2017}
Zhang, J., Chen, Q., Lin, B., and Zhou, Y. (2017).
\newblock On the single-index model estimate of the conditional density
  function: Consistency and implementation.
\newblock {\em Journal of Statistical Planning and Inference}, 187:56--66.

\bibitem[Zimmerman et~al., 2006]{Zimmerman2006}
Zimmerman, J.~E., Kramer, A.~A., McNair, D.~S., Malila, F.~M., and Shaffer,
  V.~L. (2006).
\newblock Intensive care unit length of stay: Benchmarking based on acute
  physiology and chronic health evaluation ({APACHE}) {IV}.
\newblock {\em Critical care medicine}, 34:2517--2529.

\bibitem[Zou and Zhu, 2014]{Zou2014}
Zou, Q. and Zhu, Z. (2014).
\newblock M-estimators for single-index model using {B}-spline.
\newblock {\em Metrika}, 77:225--246.

\end{thebibliography}

\section{Supplementary Material} \label{app:supplement}
\beginsupplement

\subsection{Model selection}
All models in the data application (Section 6 in the article) have been fine-tuned, and different model variants were evaluated via out-of-sample predictions on the part of the data left for model selection. Table \ref{tab:robustness} and Table \ref{tab:bagging} provide detailed results and show that the performance of the methods is robust in terms of CRPS-ranking and consistent with the findings in the article. The key steps in model tuning are summarized below.

\medskip
\textbf{Response transformations:} The outcome variable, LoS, is strongly right skewed, which suggests a log-transformation.
\begin{itemize}
\item The index estimation for DIM may benefit from response transformations, but the transformation does not directly have an impact on the estimation of the conditional CDFs. Index models with (lognormal, scaled-t) and without (gamma) log-transformation of the LoS have been considered, c.f.~Section 6.2 in the article.
\item Cox regression is invariant under strictly isotonic transformations of the response, so no response transformations need to be considered.
\item Quantile regression is more robust to outliers than regression models for the mean, and it does not necessarily require transformations with skew response variables.
Nevertheless, we verified if transformation $y \mapsto \log(y + 1)$ as used in the DIM index estimation improves the results. (The transformation
$\log(y)$ was also checked but clearly inferior.) The transformed model gave only a
minor improvement on average over the ICUs, and diverging, meaningless distributions for some patients (removed for the computation of the 
averages in Table \ref{tab:robustness}), and has therefore been discarded.
\end{itemize}

\textbf{Covariate selection:} The choice of covariates, including modelling effects of continuous variables with splines, can be
expected to have similar effects for all methods.
\begin{itemize}
\item In all models, cubic splines were used to model the effects of the continuous variables age, NEMS and SAPS II.
For Cox regression and for the index in DIM, determining a suitable dimension of the spline basis was done by using
{\tt k.check} of the {\tt mgcv} package and by graphical tools for checking the robustness of the fit. The
dimension parameter {\tt k} was finally fixed at 12 for both regression methods.
\item For quantile regression, cubic splines with equispaced knots or with knots at quantiles of the respective variables
in the training data have been compared. The equispaced knots yielded better results, with a spline space dimension similar to
the one for DIM and Cox regression. Additive quantile regression smoothing ({\tt rqss}) in the {\tt quantreg} package has
been explored, but it only offers estimation at single quantiles for each fit and up to two continuous covariates, so 
the standard method {\tt rq} has been selected.
\item We have explored whether merging factor levels with few observations (less than 30 or 50 per category)
improves the predictions. The effect was clearly negative for point forecasts for the mean LoS, as judged with the coefficient
of predictive ability \citep{Gneiting2019}, and has not been further pursued.
\end{itemize}

\textbf{Model-specific parameters:}

\begin{itemize}
\item DIM: The influence of different parametric families for the index function is discussed in Section 6.2 in the article, see also Table \ref{tab:robustness}. Detailed results on the CRPS differences with and without bagging are in Table \ref{tab:bagging}.
\item Cox regression: A possibility to make Cox regression more flexible is stratification by categorical variables. We did not pursue this approach
because it may drastically reduce the number of observations for the baseline hazard estimation and thus for the CDF estimation for some groups of 
patients. (This may be less a problem if the hazard rate and not the distribution functions are the object of interest.)
\item Quantile regression: Quantile regression is estimated on a grid of quantiles. As mentioned in the article, grids with spacing of $0.01$ and $0.001$ have been compared. In principle, the function {\tt rq} in the {\tt quantreg} package offers estimation of the full quantile regression process, but the resulting grid was too fine and led to computational difficulties. As can be seen by comparing the sixth and seventh column in Table \ref{tab:robustness}, a finer grid consistently reduces the CRPS over the ICUs. But given that the improvement by moving from a spacing of 0.01 to 0.001 is rather small, we expect only minor benefits from estimating the whole quantile regression process. 
\end{itemize}

\begin{table}
\centering
\caption{Mean CRPS on data for model selection for different variants of distributional regression methods. Asterisks ($^*$, $^{**}$, $^{***}$) indicate the
three models with the lowest CRPS for each ICU. The DIM models are abbreviated as {\tt logn}, {\tt scat} and {\tt gamma} for the variants with lognormal, scaled-t and
gamma parametric families for index estimation, without bagging. The codes for quantile regression represent models with equispaced knots for splines ({\tt e}) or
with knots at quantiles of the respective variables ({\tt q}), untransformed response variable ({\tt u}) or with the transformation
$\log(1 + y)$ ({\tt log}). The first quantile regression model (sixth column in table) is fitted on a grid with spacing 0.001 ({\tt 0.001}), the others on 
a grid with spacing 0.01. \label{tab:robustness}}
\bigskip
\resizebox{0.95\linewidth}{!}{%
\begin{tabular}{l|l|lll|lllll}
\toprule
 & Cox.~reg. & \multicolumn{3}{c|}{DIM} & \multicolumn{5}{c}{Quantile regression} \\
\midrule
Tuning &			 & {\tt logn} & {\tt scat} & {\tt gamma} & {\tt e\_u\_0.001} & {\tt e\_u} & {\tt e\_log} & {\tt q\_log} & {\tt q\_u}  \\
\midrule
ICU4 & $1.212^{}$ & $1.188^{*}$ & $1.193^{***}$ & $1.190^{**}$ & $1.202^{}$ & $1.203^{}$ & $1.200^{}$ & $1.203^{}$ & $1.204^{}$ \\ 
ICU6 & $1.632^{}$ & $1.606^{*}$ & $1.610^{**}$ & $1.622^{***}$ & $1.628^{}$ & $1.631^{}$ & $1.623^{}$ & $1.628^{}$ & $1.644^{}$ \\ 
ICU10 & $1.094^{}$ & $1.076^{**}$ & $1.081^{***}$ & $1.075^{*}$ & $1.098^{}$ & $1.099^{}$ & $1.090^{}$ & $1.092^{}$ & $1.103^{}$ \\ 
ICU19 & $1.253^{}$ & $1.248^{}$ & $1.252^{}$ & $1.262^{}$ & $1.241^{***}$ & $1.242^{}$ & $1.238^{**}$ & $1.237^{*}$ & $1.243^{}$ \\ 
ICU20 & $1.880^{}$ & $1.839^{*}$ & $1.865^{***}$ & $1.853^{**}$ & $1.904^{}$ & $1.908^{}$ & $1.885^{}$ & $1.882^{}$ & $1.917^{}$ \\ 
ICU24 & $0.972^{}$ & $0.937^{*}$ & $0.946^{}$ & $0.945^{***}$ & $0.948^{}$ & $0.948^{}$ & $0.951^{}$ & $0.945^{**}$ & $0.955^{}$ \\ 
ICU33 & $0.903^{}$ & $0.895^{}$ & $0.897^{}$ & $0.897^{}$ & $0.893^{**}$ & $0.894^{}$ & $0.893^{*}$ & $0.893^{***}$ & $0.895^{}$ \\ 
ICU39 & $1.907^{}$ & $1.865^{*}$ & $1.879^{***}$ & $1.870^{**}$ & $1.884^{}$ & $1.885^{}$ & $1.883^{}$ & $1.885^{}$ & $1.891^{}$ \\ 
ICU44 & $2.266^{}$ & $2.232^{*}$ & $2.239^{***}$ & $2.238^{**}$ & $2.298^{}$ & $2.301^{}$ & $2.263^{}$ & $2.263^{}$ & $2.307^{}$ \\ 
ICU47 & $1.306^{}$ & $1.233^{}$ & $1.255^{}$ & $1.245^{}$ & $1.220^{*}$ & $1.221^{**}$ & $1.234^{}$ & $1.227^{***}$ & $1.232^{}$ \\ 
ICU52 & $2.034^{}$ & $1.998^{*}$ & $1.999^{**}$ & $2.012^{}$ & $2.002^{***}$ & $2.004^{}$ & $2.010^{}$ & $2.011^{}$ & $2.007^{}$ \\ 
ICU55 & $1.196^{}$ & $1.178^{*}$ & $1.210^{}$ & $1.187^{}$ & $1.184^{}$ & $1.185^{}$ & $1.182^{**}$ & $1.182^{***}$ & $1.186^{}$ \\ 
ICU58 & $1.344^{}$ & $1.312^{*}$ & $1.317^{**}$ & $1.320^{***}$ & $1.329^{}$ & $1.330^{}$ & $1.344^{}$ & $1.330^{}$ & $1.328^{}$ \\ 
ICU65 & $1.069^{}$ & $1.004^{*}$ & $1.007^{**}$ & $1.010^{***}$ & $1.011^{}$ & $1.012^{}$ & $1.029^{}$ & $1.040^{}$ & $1.024^{}$ \\ 
ICU76 & $2.552^{}$ & $2.521^{**}$ & $2.532^{***}$ & $2.517^{*}$ & $2.551^{}$ & $2.552^{}$ & $2.543^{}$ & $2.549^{}$ & $2.558^{}$ \\ 
ICU77 & $0.838^{}$ & $0.832^{**}$ & $0.835^{***}$ & $0.825^{*}$ & $0.842^{}$ & $0.843^{}$ & $0.837^{}$ & $0.837^{}$ & $0.845^{}$ \\ 
ICU79 & $1.266^{}$ & $1.211^{*}$ & $1.215^{**}$ & $1.233^{***}$ & $1.263^{}$ & $1.263^{}$ & $1.267^{}$ & $1.285^{}$ & $1.257^{}$ \\ 
ICU80 & $0.996^{}$ & $0.983^{**}$ & $0.999^{}$ & $0.981^{*}$ & $0.998^{}$ & $0.998^{}$ & $0.992^{***}$ & $0.997^{}$ & $1.002^{}$ \\ 
\midrule
Mean & $1.429^{}$ & $1.398^{*}$ & $1.407^{***}$ & $1.405^{**}$ & $1.416^{}$ & $1.418^{}$ & $1.415^{}$ & $1.416^{}$ & $1.422^{}$ \\ 
\bottomrule
\end{tabular}%
}
\end{table}

\begin{table}
\centering
\caption{Increase in CRPS of the DIM when in-sample predictions on the training data are used for the estimation of the conditional CDFs instead of the bagging approach with 100 subsamples, for the lognormal, scaled-t and gamma index models. See Table \ref{tab:robustness} for the average CRPS without bagging. Positive values correspond to higher CRPS (worse predictions) of the variant without bagging. \label{tab:bagging}}
\bigskip
\begin{tabular}{l|ccc}
\toprule
      & Lognormal & Scaled-t  & Gamma \\ 
\midrule
ICU4 & $0.0030$ & $-0.001$ & $0.0020$ \\ 
ICU6 & $0.0060$ & $0.0050$ & $0.0130$ \\ 
ICU10 & $0.0010$ & $0.0010$ & $0.0010$ \\ 
ICU19 & $0.0100$ & $0.0060$ & $0.0230$ \\ 
ICU20 & $0.0030$ & $0.0240$ & $0.0160$ \\ 
ICU24 & $0.0020$ & $0.0040$ & $0.0050$ \\ 
ICU33 & $0.0040$ & $0.0010$ & $0.0030$ \\ 
ICU39 & $0.0100$ & $0.0070$ & $0.0120$ \\ 
ICU44 & $-0.002$ & $-0.002$ & $0.0050$ \\ 
ICU47 & $0.0030$ & $0.0020$ & $0.0110$ \\ 
ICU52 & $0.0070$ & $0.0040$ & $0.0060$ \\ 
ICU55 & $0.0020$ & $0.0190$ & $0.0150$ \\ 
ICU58 & $0.0060$ & $0.0040$ & $0.0100$ \\ 
ICU65 & $0.0040$ & $0.0030$ & $0.0030$ \\ 
ICU76 & $000000$ & $0.0030$ & $000000$ \\ 
ICU77 & $0.0070$ & $0.0070$ & $0.0020$ \\ 
ICU79 & $0.0070$ & $-0.001$ & $0.0110$ \\ 
ICU80 & $0.0040$ & $0.0080$ & $0.0020$ \\ 
\midrule
Mean & $0.0040$ & $0.0050$ & $0.0080$ \\ 
\bottomrule
\end{tabular}
\end{table}

\subsection{Discreteness of LoS}
\citet[Appendix SB]{Chernozhukov2013} demonstrate that discreteness in the outcome variable influences the performance of quantile regression relative to other distributional regression techniques. Table \ref{tab:outcome_discrete} and Figure \ref{fig:outcome_discrete} summarize the cumulative proportion
of the most frequent LoS values for each ICU as a measure of discreteness. Compared with Figure SB.1.~in the supplementary material of \cite{Chernozhukov2013}, the discreteness
in the outcome variable is substantially lower in our study. Moreover, there is no relationship between the 
performance of DIM and Cox regression relative to quantile regression (Table 3 in the article) and the degree of
discreteness as summarized in Table \ref{tab:outcome_discrete}. 
As mentioned in the first paragraph of Section 6.4 and visible in Figure 3 in the article, quantile regression
indeed has difficulties in fitting the pattern in the ICU discharge times with marked peaks before noon and in the afternoon.
Nevertheless, it clearly outperforms Cox regression, which is able to correctly recognize this pattern. Based on these two observations, we argue that the disadvantage of quantile regression due to discreteness of the LoS is at most of limited extent and not decisive in our study.

\begin{figure}
\centering
\caption{Cumulative probabilities of LoS attaining one of the $k$ most frequent values, $k = 1, 2, \dots, 25$, stratified
by ICU (identifiers omitted). \label{fig:outcome_discrete}}
\bigskip
\includegraphics[width = 0.8\textwidth]{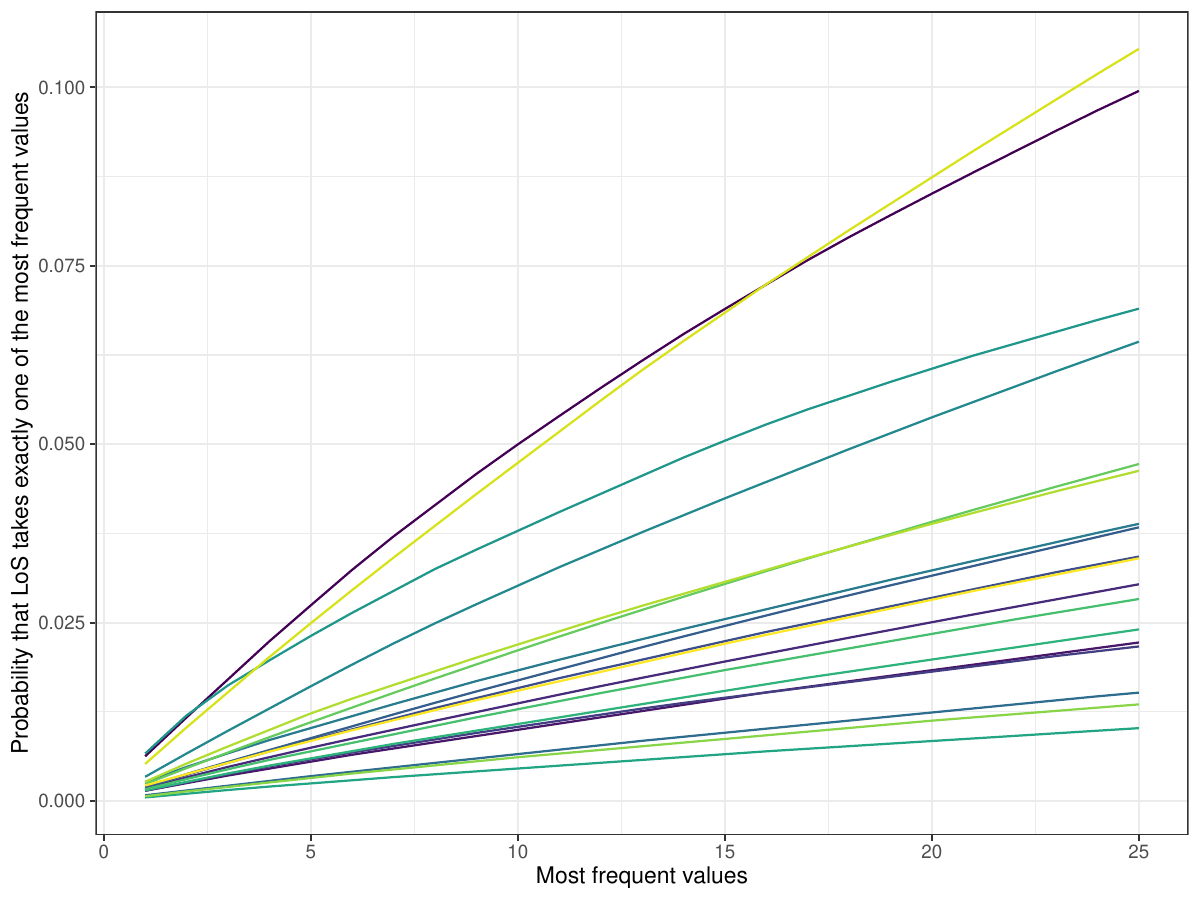}
\end{figure}

\begin{table}
\centering
\caption{Cumulative probabilities of LoS attaining one of the $k$ most frequent values, $k = 1, \, 2,\, 10,\, 25$, stratified
by ICU. \label{tab:outcome_discrete}}
\bigskip
\begin{tabular}{l|r|r|r|r|r|r|r}
\toprule
ICU & 1 & 2 & 3 & 4 & 5 & 10 & 25\\
\midrule
ICU4 & 0.006 & 0.012 & 0.017 & 0.022 & 0.027 & 0.050 & 0.099\\
ICU6 & 0.001 & 0.003 & 0.004 & 0.005 & 0.006 & 0.010 & 0.022\\
ICU10 & 0.002 & 0.003 & 0.005 & 0.006 & 0.007 & 0.014 & 0.030\\
ICU19 & 0.001 & 0.003 & 0.004 & 0.005 & 0.006 & 0.010 & 0.022\\
ICU20 & 0.002 & 0.004 & 0.005 & 0.007 & 0.009 & 0.016 & 0.034\\
ICU24 & 0.002 & 0.004 & 0.005 & 0.007 & 0.009 & 0.017 & 0.038\\
ICU33 & 0.001 & 0.001 & 0.002 & 0.003 & 0.004 & 0.007 & 0.015\\
ICU39 & 0.003 & 0.005 & 0.007 & 0.009 & 0.010 & 0.018 & 0.039\\
ICU44 & 0.003 & 0.007 & 0.010 & 0.013 & 0.016 & 0.030 & 0.064\\
ICU47 & 0.007 & 0.012 & 0.016 & 0.020 & 0.023 & 0.038 & 0.069\\
ICU52 & 0.001 & 0.001 & 0.002 & 0.002 & 0.002 & 0.005 & 0.010\\
ICU55 & 0.002 & 0.003 & 0.004 & 0.005 & 0.006 & 0.011 & 0.024\\
ICU58 & 0.002 & 0.003 & 0.004 & 0.006 & 0.007 & 0.013 & 0.028\\
ICU65 & 0.002 & 0.005 & 0.007 & 0.009 & 0.011 & 0.021 & 0.047\\
ICU76 & 0.001 & 0.001 & 0.002 & 0.003 & 0.003 & 0.006 & 0.014\\
ICU77 & 0.003 & 0.005 & 0.008 & 0.010 & 0.012 & 0.022 & 0.046\\
ICU79 & 0.005 & 0.010 & 0.015 & 0.020 & 0.025 & 0.047 & 0.105\\
ICU80 & 0.002 & 0.004 & 0.005 & 0.007 & 0.008 & 0.015 & 0.034\\
\bottomrule
\end{tabular}
\end{table}

\subsection{Additional figures and tables}

Table \ref{tab:ssummaries} shows summary statistics of the ICU LoS, patient age, SAPS II and NEMS for all ICUs.

\bigskip
Figure \ref{fig:checkRqCrossings} shows probabilistic LoS forecasts obtained by quantile regression, for eight randomly selected patients per ICU. While there are some crossings in the CDFs (e.g.~in ICUs 47 and 52), the CDFs for most patients do not cross and are hence comparable with respect to stochastic dominance, suggesting that the model assumption of the DIM is reasonable for ICU LoS.

\bigskip
Figures \ref{fig:reliability17_revised} and \ref{fig:reliability814_revised} show reliability diagrams for the predicted probability that the LoS exceeds $k = 1, \, 2, \, \ldots, 14$ days for all forecasting methods and ICUs. PIT histograms are shown in Figures \ref{fig:pit1_revised} and \ref{fig:pit2_revised}.

\bigskip
Figure \ref{fig:crpsDiff_revised} shows the difference in CRPS between the quantile regression forecasts and the DIM forecasts. For all ICUs, there is a considerable number of outliers (defined as points outside the $25\%$ ($75\%$) quantile minus (plus) $1.5$ times the interquartile range), so Wilcoxon's signed rank test was applied to compare the CRPS, instead of a t-test.

\begin{table}
\caption{Summary statistics (mean, median and standard deviation) of numeric variables in the dataset. \label{tab:ssummaries}} 
\bigskip
\resizebox{\linewidth}{!}{\begin{tabular}{l|r|r|r||r|r|r||r|r|r||r|r|r}
\toprule
ICU identifier & \multicolumn{3}{c||}{LoS} & \multicolumn{3}{c||}{Age} & \multicolumn{3}{c||}{NEMS} & \multicolumn{3}{c}{SAPS}\\
\hline\hline
& mean & median & sd & mean & median & sd & mean & median & sd & mean & median & sd \\
\hline
ICU4 & 2.1 & 0.7 & 4.4 & 65.7 & 68 & 14.5 & 25.4 & 25 & 9.6 & 29.0 & 26 & 14.5\\
\hline
ICU6 & 2.8 & 0.8 & 5.5 & 65.2 & 69 & 16.7 & 23.3 & 21 & 7.9 & 35.8 & 33 & 18.1\\
\hline
ICU10 & 2.0 & 0.7 & 3.9 & 62.9 & 66 & 15.7 & 27.8 & 27 & 8.8 & 41.6 & 39 & 18.2\\
\hline
ICU19 & 2.1 & 0.9 & 4.7 & 64.7 & 67 & 15.3 & 20.1 & 18 & 7.4 & 29.9 & 25 & 16.6\\
\hline
ICU20 & 3.7 & 0.7 & 8.0 & 64.2 & 67 & 15.3 & 25.5 & 24 & 8.3 & 31.6 & 27 & 17.5\\
\hline
ICU24 & 2.0 & 0.6 & 4.5 & 63.4 & 66 & 15.4 & 24.1 & 18 & 7.7 & 29.5 & 26 & 16.8\\
\hline
ICU33 & 2.0 & 1.0 & 3.3 & 66.1 & 69 & 15.8 & 19.9 & 18 & 7.5 & 36.5 & 33 & 17.4\\
\hline
ICU39 & 2.9 & 1.0 & 6.2 & 62.6 & 65 & 16.5 & 23.2 & 18 & 7.1 & 28.8 & 26 & 15.9\\
\hline
ICU44 & 3.9 & 1.5 & 7.8 & 59.0 & 61 & 17.6 & 27.1 & 27 & 8.5 & 34.0 & 31 & 18.9\\
\hline
ICU47 & 2.5 & 1.5 & 5.1 & 67.6 & 69 & 12.8 & 25.9 & 25 & 7.4 & 27.7 & 26 & 12.7\\
\hline
ICU52 & 3.7 & 1.6 & 6.3 & 60.5 & 63 & 17.3 & 26.2 & 27 & 10.3 & 40.8 & 39 & 18.5\\
\hline
ICU55 & 2.4 & 0.8 & 4.4 & 64.6 & 67 & 16.1 & 20.6 & 18 & 7.8 & 30.8 & 27 & 16.4\\
\hline
ICU58 & 2.6 & 0.7 & 4.8 & 61.7 & 64 & 16.4 & 22.5 & 18 & 7.3 & 28.5 & 26 & 15.0\\
\hline
ICU65 & 1.8 & 0.6 & 4.3 & 67.2 & 69 & 13.9 & 25.5 & 25 & 7.9 & 28.7 & 28 & 12.5\\
\hline
ICU76 & 4.3 & 1.7 & 7.2 & 63.2 & 66 & 15.6 & 30.3 & 30 & 8.3 & 41.2 & 40 & 17.2\\
\hline
ICU77 & 1.8 & 0.6 & 3.2 & 65.0 & 68 & 15.9 & 21.9 & 18 & 8.0 & 31.1 & 28 & 16.1\\
\hline
ICU79 & 2.7 & 0.5 & 5.9 & 55.8 & 57 & 17.0 & 22.4 & 18 & 7.1 & 19.1 & 15 & 15.3\\
\hline
ICU80 & 1.8 & 0.6 & 3.7 & 65.3 & 68 & 16.1 & 19.4 & 18 & 7.3 & 29.0 & 27 & 13.1\\
\bottomrule
\end{tabular}}
\end{table}

\begin{figure}
\caption{Predictive CDFs obtained by quantile regression, for randomly selected patients. \label{fig:checkRqCrossings}}
\bigskip
\center
\includegraphics[width = 0.9\textwidth]{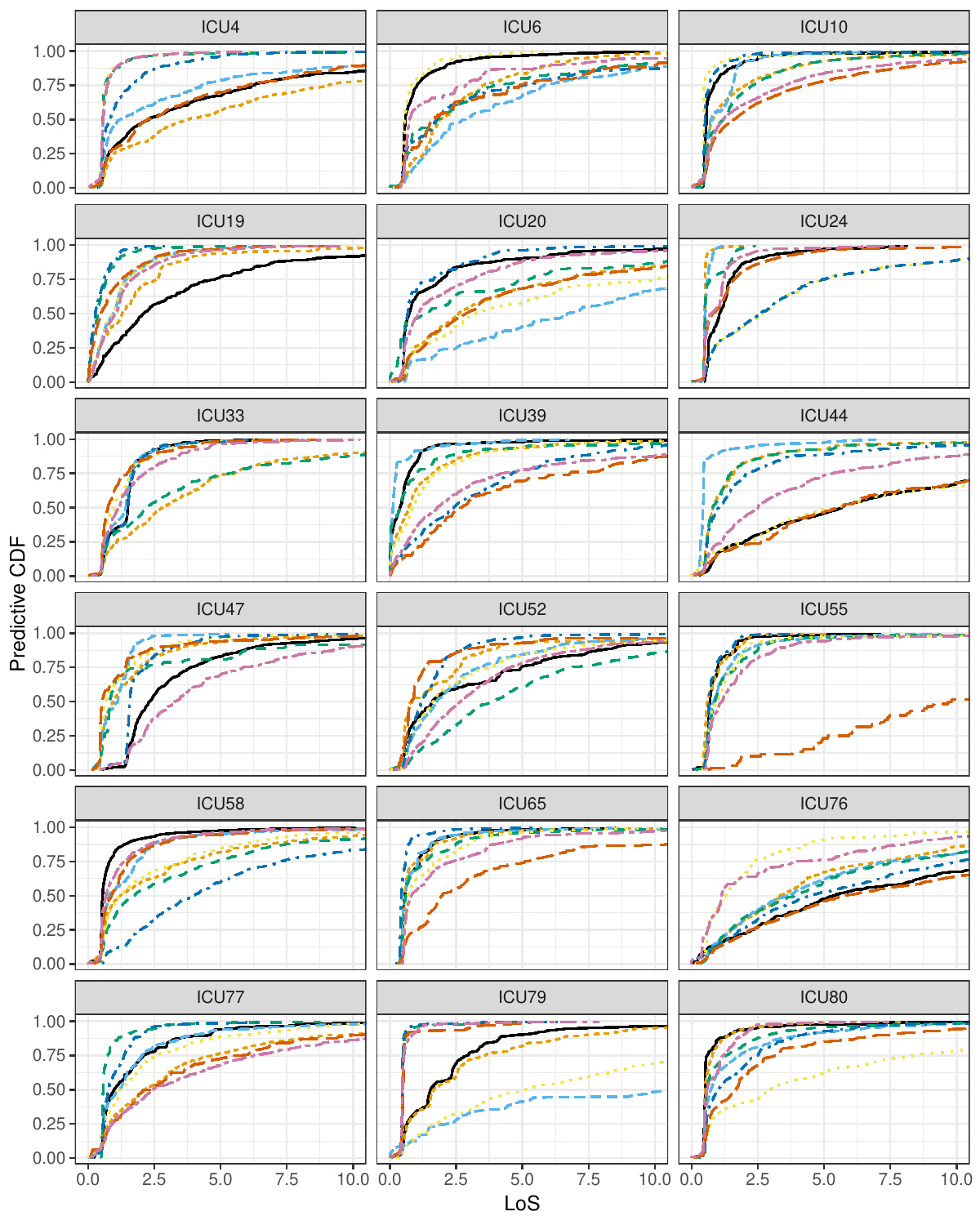}
\end{figure}

\begin{figure}
\caption{Reliability diagrams of probabilistic forecasts for the predicted probability that the LoS exceeds $1, \, 2, \, \ldots, 7$ days. The forecast probability is grouped into the bins $[0, 0.1], (0.1, 0.2], \ldots, (0.9, 1]$. Only bins with more than two observations are included. \label{fig:reliability17_revised}}
\bigskip
\center
\includegraphics[width = 0.8\textwidth]{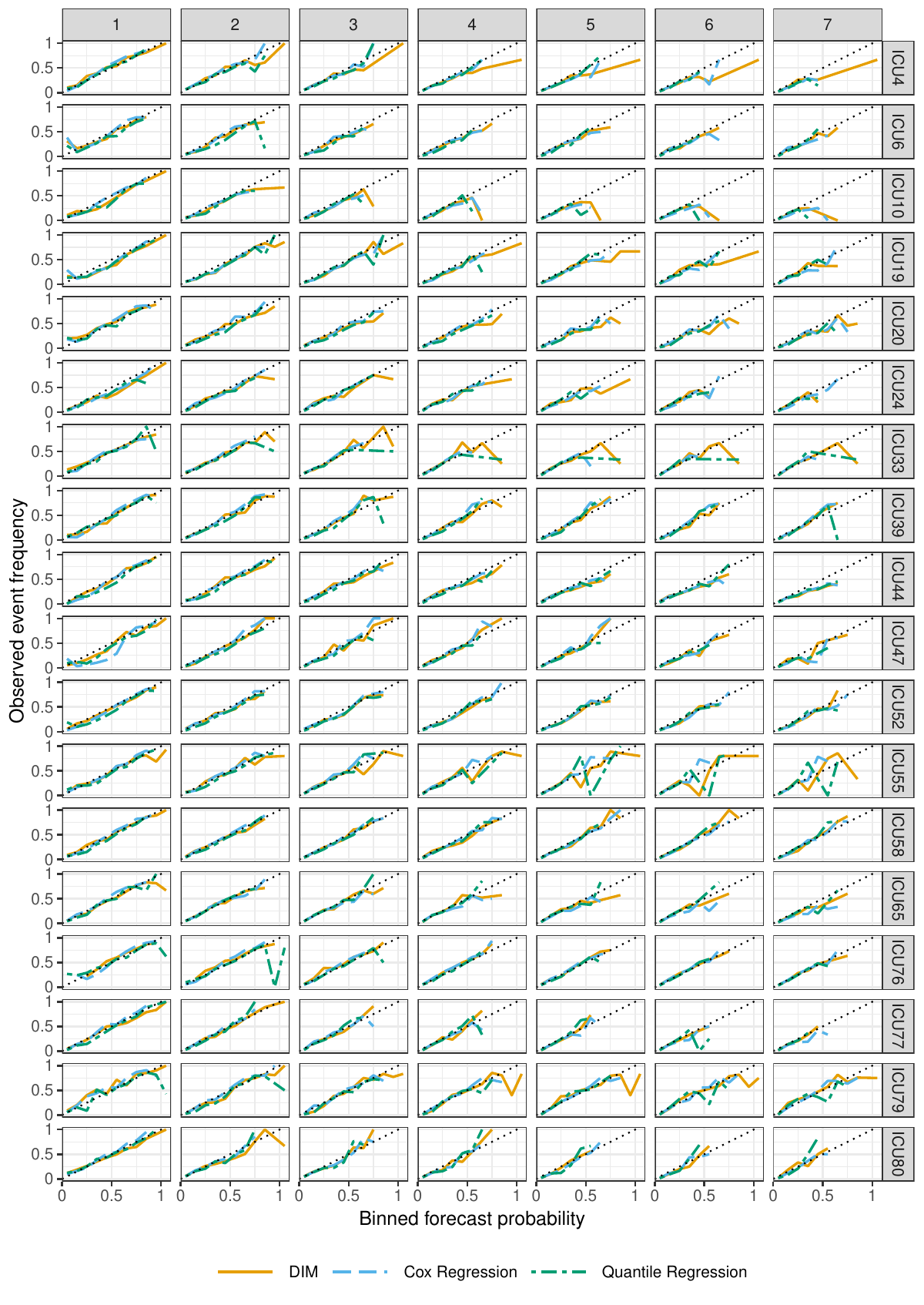}
\end{figure}

\begin{figure}
\caption{Reliability diagrams of probabilistic forecasts for the predicted probability that the LoS exceeds $8, 9, \ldots, 14$ days. The curves are as specified in Figure \ref{fig:reliability17_revised}. \label{fig:reliability814_revised}}
\bigskip
\center
\includegraphics[width = 0.8\textwidth]{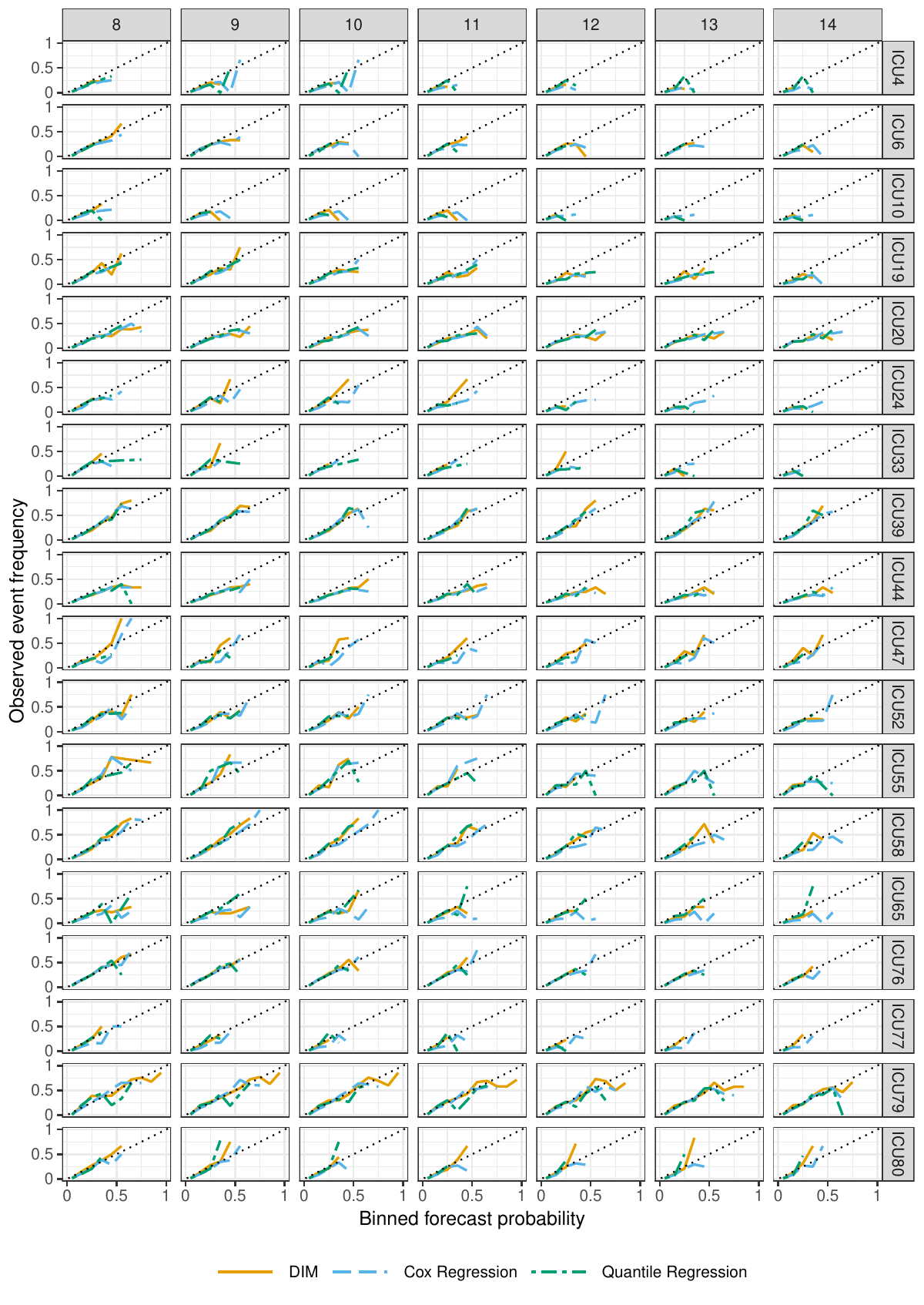}
\end{figure}

\begin{figure}
\caption{PIT histograms of the probabilistic forecasts with bins of width $1/20$ (first nine ICUs). \label{fig:pit1_revised}}
\bigskip
\center
\includegraphics[width = 0.7\textwidth]{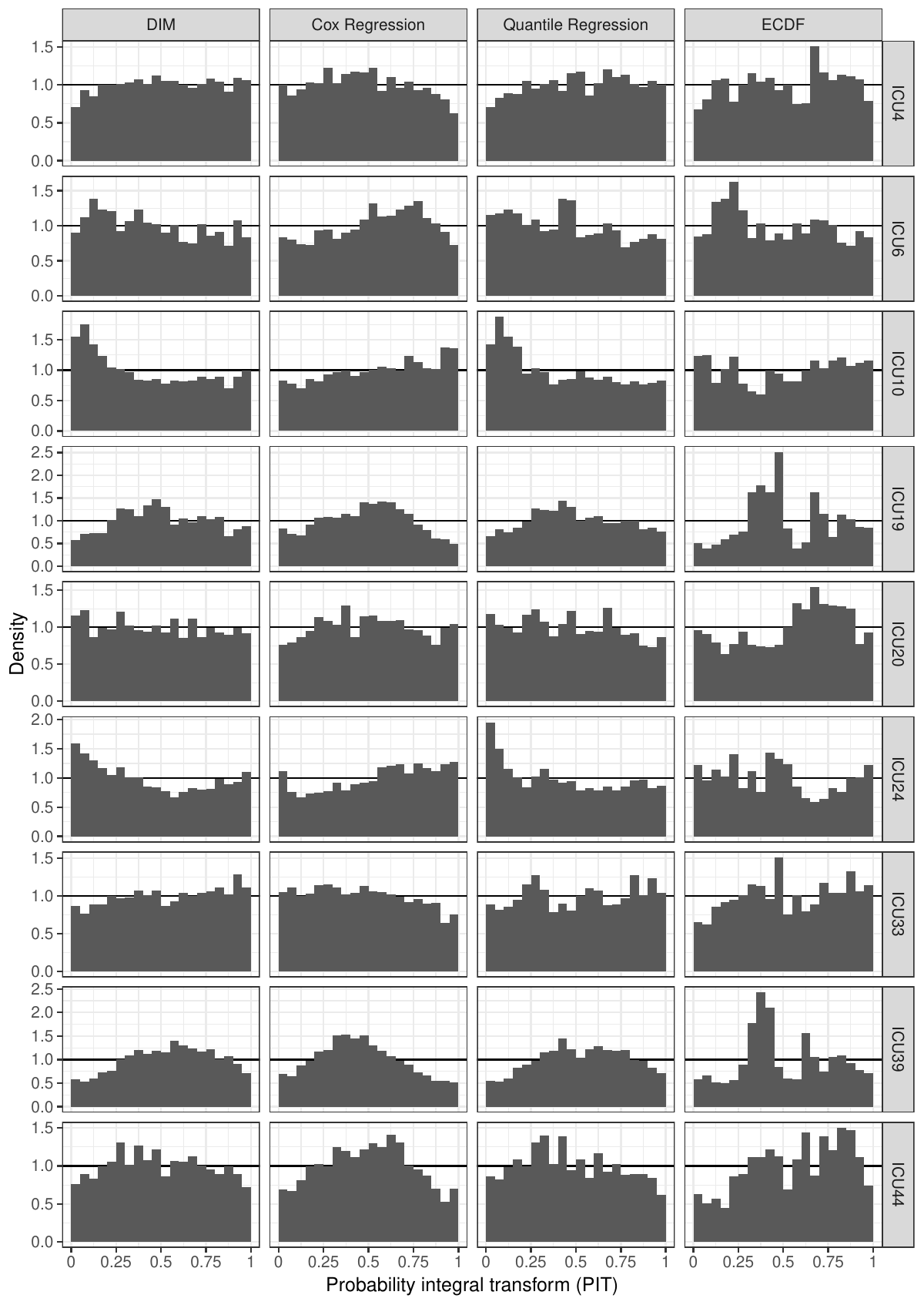}
\end{figure}

\begin{figure}
\caption{PIT histograms of the probabilistic forecasts with bins of width $1/20$ (second half of the ICUs). \label{fig:pit2_revised}}
\bigskip
\center
\includegraphics[width = 0.7\textwidth]{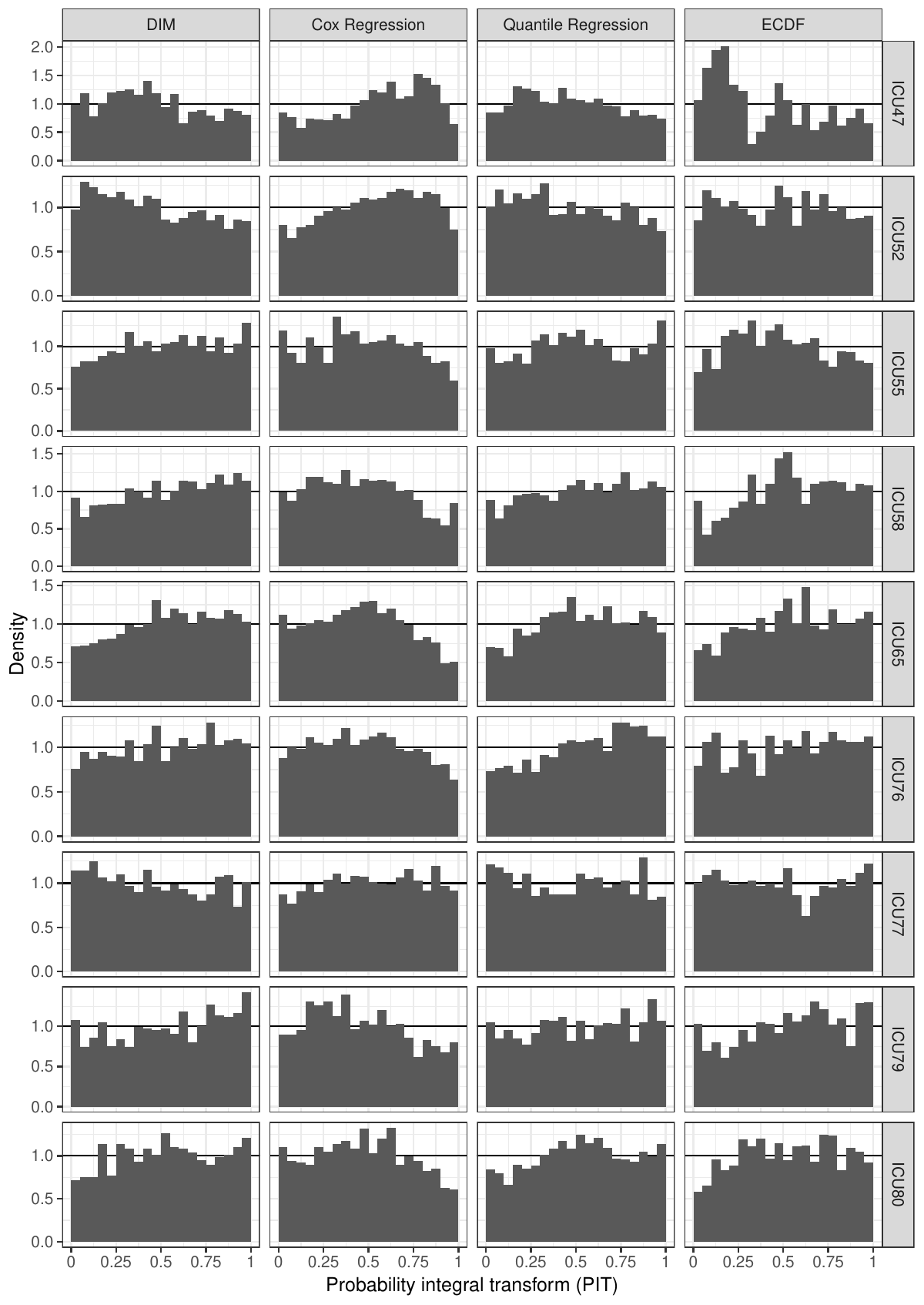}
\end{figure}

\begin{figure}
\caption{Boxplot of the difference in the CRPS of the quantile regression forecasts and of the DIM forecasts. Outliers are displayed as crosses (with horizontal jitter). \label{fig:crpsDiff_revised}}
\bigskip
\center
\includegraphics[width = 0.7\textwidth]{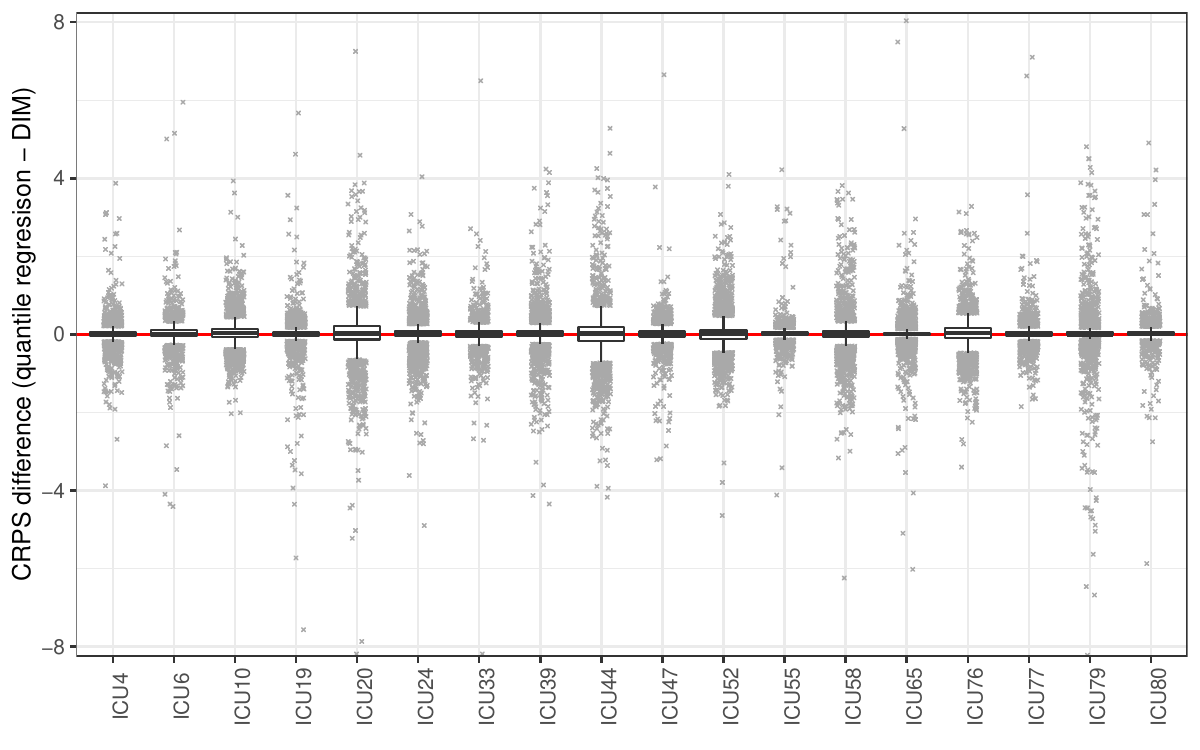}
\end{figure}

\end{document}